%% file: main.tex
\documentclass[11pt]{article}
\usepackage[top=1in,left=1in,right=1in,bottom=1in]{geometry}

\usepackage{amsthm,comment,caption}
\newtheorem{lemma}{Lemma}[section]
\newtheorem{corollary}{Corollary}[section]

\usepackage{listings,multicol,lipsum}
\usepackage{style}
\usepackage{enumitem}
\usepackage{subcaption}
\usepackage{balance}  

\input{macro}

\usepackage{times,microtype}

\begin{document}
\date{}

\title{The Parallel Persistent Memory Model \footnote{This paper is the full version of a paper at SPAA 2018 with the same name.}}

\author{
  Guy E. Blelloch\\Carnegie Mellon University\\guyb@cs.cmu.edu \and
  Phillip B. Gibbons\\Carnegie Mellon University\\gibbons@cs.cmu.edu \and
  Yan Gu\\Carnegie Mellon University\\yan.gu@cs.cmu.edu \and
  Charles McGuffey\\Carnegie Mellon University\\cmcguffe@cs.cmu.edu \and
  Julian Shun\\MIT CSAIL\\jshun@mit.edu}

\maketitle



\newcommand{\locmem}{ephemeral memory}
\newcommand{\permem}{persistent memory}
\newcommand{\ourmodel}{PM}
\newcommand{\ourparallelmodel}{Parallel-PM}
\newcommand{\faultprob}{f}
\newcommand{\idem}{idempotent}
\newcommand{\Wf}{W\!\!_f}
\newcommand{\war}{write-after-read}

\input{abstract}

\input{intro}


\input{model}

\input{single-proc-robust}

\input{prog-robust}

\input{multprocs}

\input{work-stealing}

\input{algorithms}

\input{conclusion}

\section*{Acknowledgements}
This work was supported in part by NSF grants CCF-1408940, CCF-1533858, and CCF-1629444.

\bibliographystyle{abbrv}

\input{main.bbl}

\begin{appendix}
  \input{work-stealing-formal}
\end{appendix}

\end{document}

%% file: macro.tex
\newcommand{\wcost}{\omega}

\newcommand{\id}[1]{\ifmmode\mathit{#1}\else\textit{#1}\fi}
\newcommand{\const}[1]{\ifmmode\mbox{\textc{#1}}\else\textsc{#1}\fi}


%% file: abstract.tex
\begin{abstract}

We consider a parallel computational model, the \emph{Parallel
  Persistent Memory model}, comprised of $P$ processors, each with a
fast local ephemeral memory of limited size, and sharing a large
persistent memory.  The model allows for each processor to fault at
any time (with bounded probability), and possibly restart.  When a
processor faults, all of its state and local ephemeral memory is lost,
but the persistent memory remains.  This model is motivated by
upcoming non-volatile memories that are nearly as fast as existing
random access memory, are accessible at the granularity of cache
lines, and have the capability of surviving power outages.  It is
further motivated by the observation that in large parallel systems,
failure of processors and their caches is not unusual.

We present several results for the model, using an approach that
breaks a computation into \emph{capsules}, each of which can be safely
run multiple times.  For the single-processor version we describe how
to simulate any program in the RAM, the external memory model, or the
ideal cache model with an expected constant factor overhead.  For the
multiprocessor version we describe how to efficiently implement a
work-stealing scheduler within the model such that it handles both
soft faults, with a processor restarting, and hard faults, with a
processor permanently failing.  For any multithreaded fork-join
computation that is race free, \war{} conflict free and has $W$ work,
$D$ depth, and $C$ maximum capsule work in the absence of faults, the
scheduler guarantees a time bound on the model of
$O\left(\frac{W}{P_A} + \frac{DP}{P_A}
\left\lceil\log_{1/(C\faultprob)} W\right\rceil\right)$ in
expectation, where $P$ is the maximum number of processors, $P_A$ is
the average number, and $\faultprob{} \leq 1/(2C)$ is the probability
a processor faults between successive persistent memory accesses.
Within the model, and using the proposed methods, we develop efficient
algorithms for parallel prefix sums, merging, sorting, and matrix
multiply.

\end{abstract}

%% file: intro.tex
\section{Introduction}

In this paper, we consider a parallel computational model, the
\emph{Parallel Persistent Memory (\ourparallelmodel) model}, 
that consists of $P$ processors, each with a fast local
ephemeral memory of limited size $M$, and sharing a large slower
persistent memory.  As in the external memory
model~\cite{ABP01,Arge08}, each processor runs a standard instruction
set from its \locmem{} and has instructions for transferring blocks of
size $B$ to and from the persistent memory.  The cost of an algorithm
is calculated based on the number of such transfers.  A key
difference, however, is that the model allows for individual processors to
fault at any time.  If a processor faults, all of its processor
state and local ephemeral memory is lost, but the persistent memory
remains.  We consider both the case
where the processor restarts (soft faults) and the case where it never restarts (hard faults).  

The model is motivated by two complimentary trends.  Firstly, it is
motivated by upcoming non-volatile memories that are nearly as fast as
existing random access memory (DRAM), are accessed via loads and stores at
the granularity of cache lines, have large capacity (more bits per
unit area than existing random access memory), and have the capability
of surviving power outages and other failures without losing data (the
memory is \emph{non-volatile} or \emph{persistent}).  For example,
Intel's 3D-Xpoint memory technology, currently available as an SSD, is
scheduled to be available as such a random access memory in 2019.
While such memories are expected to be the pervasive type of 
memory~\cite{Nawab2017,Meena14,Yole13}, 
each processor will still have a small amount of cache and other fast memory
implemented with traditional \emph{volatile} memory technologies (SRAM or DRAM).
Secondly, it is motivated by the fact that in current and upcoming
large parallel systems the probability that an individual processor
faults is not negligible, requiring some form of fault
tolerance~\cite{Cappello14}.  

In this paper, we first consider a single processor version of the
model, the \emph{\ourmodel{} model}, and give conditions under which programs are robust against
faults.  In particular, we identify that breaking a computation into
``capsules'' that have no \war{} conflicts (writing a location
that was read earlier within the same capsule) is sufficient, when combined with our
approach to restarting faulting capsules from their beginning,
due to its idempotent behavior.  We then show that RAM
algorithms, external memory algorithms, and cache-oblivious
algorithms~\cite{Frigo99} can all be implemented asymptotically efficiently on the model.  This involves a simulation that breaks
the computations into capsules and buffers writes, which are handled
in the next capsule.  However, the simulation is likely not practical.  We therefore
consider a programming methodology in which the algorithm designer can
identify capsule boundaries, and ensure that the capsules are free of
\war{} conflicts.

We then consider our multiprocessor counterpart, the
\ourparallelmodel{} described above, and consider conditions under which
programs are correct when the processors are interacting through the
shared memory.  We identify that if capsules are free of \war{}
conflicts and atomic, in a way that we define, then each capsule acts as if
it ran once despite many possible restarts.  Furthermore we identify
that a compare-and-swap (CAS) instruction is not safe in the
\ourmodel{} model, but that a compare-and-modify (CAM), which
does not see its result, is safe.

The most significant result in the paper is a work-stealing scheduler
that can be used on the \ourparallelmodel{}.  Our scheduler is based on
the scheduler of Arora, Blumofe, and Plaxton (ABP)~\cite{ABP01}.   The key
challenges in adopting it to handle faults are (i) modifying it so that it only
uses CAMs instead of CASs, (ii) ensuring that each stolen task gets executed 
despite faults, (iii) properly handling hard faults, and (iv) ensuring its
efficiency in the presence of soft or hard faults.   Without a CAS, and to avoid blocking,
handling faults requires that processors help the processor that is part way through a steal.
Handling hard faults further requires being able to steal a thread from a processor that
was part way through executing the thread. 

Based on the scheduler we show that any race-free, \war{} conflict
free multithreaded fork-join program with work $W$, depth $D$, and maximum capsule work $C$ will run in expected
time: 
\[O\left(\frac{W}{P_A} + D\left(\frac{P}{P_A}\right) \left\lceil\log_{1/(C\faultprob)} W\right\rceil\right).\]  Here $P$ is the
maximum number of processors, $P_A$ the average number, and
$\faultprob \leq 1/(2C)$ an upper bound on the probability 
a processor faults between successive persistent memory accesses.   This bound differs from the ABP result only in the $\log_{1/(C\faultprob)} W$ factor
on the depth term, due to faults along the critical path.

Finally, we present \ourparallelmodel{} algorithms for prefix-sums,
merging, sorting, and matrix multiply that satisfy the required
conditions.  The results for prefix-sums, merging, and sorting are
work-optimal, matching lower bounds for the external memory model.
Importantly, these algorithms are only slight modifications from
known parallel I/O efficient algorithms~\cite{blelloch2010low}.
The main change is ensuring that they write their partial results
to a separate location from where they read them so that they avoid
\war{} conflicts.

\myparagraph{Related Work} Because of its importance to future computing,
the computer systems community (including
companies such as Intel and HP) have been hard at work trying to solve
the issues arising when fast nonvolatile memories (such as caches) sit
between the processor and a large persistent
memory~\cite{Bhandari2016, deKruijf2012, Intel-isa-extensions2017,
  Izraelevitz2016, Hsu2017, Chakrabarti2014, Kolli2016, Liu2017,
  Intel-nvm-library, Colburn2011, Volos2011, Pelley2014, Chauhan2016,
  Kim2016, Memaripour2017, X10Failure17, guerraoui2004robust, lee2017wort, chen2015persistent, berryhill2016robust, Friedman18}.  Standard caches are
write-back, meaning that a write to a memory location will make it
only as far as the cache, until at some later point the updated cache
line gets flushed out to the persistent memory.  Thus, when a
processor crashes, some writes (those still in the cache) are lost
while other writes are not.  The above prior work includes schemes for
encapsulating updates to persistent memory in either
\emph{transactions} or \emph{lock-protected failure atomic sections}
and using various forms of (undo, redo, resume) logging to ensure
correct recovery.  
The intermittent computing community works on the related problem of small systems that will crash due to power loss~\cite{lucia2015simpler, colin2018termination, balsamo2015hibernus, colin2016chain, hester2017timely, van2016intermittent, buettner2011dewdrop, maeng2017alpaca}.
Lucia and Ransford~\cite{lucia2015simpler} describe how faults and restarting lead to errors that will not occur in a faultless setting. 
Several of these works~\cite{lucia2015simpler, colin2018termination, colin2016chain, van2016intermittent, maeng2017alpaca} break code into small chunks, referred to as \emph{tasks}, and work to ensure progress at that granularity. 
Avoiding \war{} conflicts is often the key step towards ensuring that tasks
are idempotent.
Because these works target intermittent computing systems, which are designed to be small and energy efficient, they do not consider multithreaded programs, concurrency, or synchronization. 
In contrast to this flurry of recent systems
research, there is
relatively little work from the theory/algorithms community
aimed at this setting~\cite{David2017, Izraelevitz2016disc,
  Izraelevitz2016spaa, Nawab2017}.  David et al.~\cite{David2017}
presents concurrent data structures (e.g., for skip-lists) that avoid
the overheads of logging.  Izraelevitz et
al.~\cite{Izraelevitz2016disc, Izraelevitz2016spaa} presents efficient
techniques for ensuring that the data in persistent memory captures a
consistent cut in the happens-before graph of the program's execution,
via the explicit use of instructions that flush cache lines to
persistent memory (such as Intel's {\tt CLFLUSH}
instruction~\cite{Intel-isa-extensions2017}).  Nawab et
al.~\cite{Nawab2017} defines \emph{periodically persistent} data
structures, which combine mechanisms for tracking proper write
ordering with a periodic flush of all cache lines to persistent
memory.
None of this work defines an algorithmic cost model, presents
a work-stealing scheduler, or provides the provable bounds in this paper.

There is a very large body of research on models and algorithms where
processors and/or memory can fault, but to our knowledge, none of it
(other than the works mentioned above) fits the setting we study with
its two classes of memory (local volatile and shared nonvolatile).
Papers focusing on memory faults (e.g., \cite{FinocchiI04, Afek92,
  Chlebus1994} among a long list of such papers) consider models in
which individual memory locations can fault.  Papers focusing on
processor faults (e.g., \cite{Aumann1991} among an even longer list of
such papers) either do not consider memory faults or assume that all
memory is volatile.

\myparagraph{Write-back Caches} Note that while the \ourmodel{} models are defined using
explicit external read and external write instructions, they are also 
appropriate for modeling the (write-back) cache setting described above,
as follows.  Explicit instructions, such as {\tt CLFLUSH}, are used to
ensure that an external write indeed writes to the persistent memory. 
Writes that are intended to be solely in local memory, on the other hand,
could end up being evicted from the cache and written back to persistent
memory.  However, for programs that are race-free and well-formed, as 
defined in Section~\ref{sec-single-proc-robust}, our approach
preserves its correctness properties.



%% file: model.tex
\section{The Persistent Memory Model}

\myparagraph{Single Processor}
We assume a two-layer memory model with a small fast \emph{\locmem{}}
of size $M$ (in words) and a large slower \emph{\permem{}} of size
$M_p \gg M$.  The two memories are partitioned into blocks of $B$
words.  Instructions include standard RAM instructions that work on
single words within the processor registers (a processor has $O(1)$
registers) and \locmem{}, as well as two \emph{(external) memory transfer} instructions:
an \emph{external read} that transfers a block from \permem{} into
\locmem{}, and an \emph{external write} that transfers a block from
\locmem{} to \permem{}.  We assume that the words contain $\Theta(\log
M_p)$ bits.  These assumptions are effectively the same as in the
$(M,B)$ external memory model~\cite{AggarwalV88}.

We further assume that the processor can \emph{fault} between any two
instructions,\footnote{For simplicity, we assume that individual instructions 
are atomic.} and that after faulting, the processor \emph{restarts}.
On restart, the \locmem{} and processor registers can be in an
arbitrary state, but the \permem{} is in the same state as
immediately before the fault. 
To enable forward progress, we assume there is a fixed memory location
in the persistent memory referred to as the \emph{restart pointer
  location}, containing a \emph{restart pointer}.  On restart,
the processor loads the restart pointer from the persistent memory
into a register, which we refer to as the \emph{base} register, and then
loads the location pointed to by the restart pointer (the
\emph{restart instruction}) and jumps to that location, i.e., sets it
as the program counter.  The processor then proceeds as normal.
As it executes, the processor can update the restart pointer to 
be the current program counter, at the cost of an external write,
in order to limit how far the processor will fall back on a fault.
We refer to this model as the
(single processor) $(M,B)$ persistent memory (\ourmodel{}) model.



The basic model can be parameterized based on the cost of the various
instructions.  
Throughout this paper, and in the spirit of the
external memory model~\cite{AggarwalV88} and the ideal cache
model~\cite{Frigo99}, we assume that external reads and writes take
unit cost and all other instructions have no cost.\footnote{The results in this paper 
can be readily extended to a setting (an \emph{Asymmetric \ourmodel{} model}) where external writes are more costly than
external reads, as in prior work on algorithms for NVM~\cite{Carsonetal15,
BBFGGMS18implicit, blelloch2016efficient, BFGGS15, BBFGGMS16, jacob17};
for simplicity, we study here the simpler \ourmodel{}
model because such asymmetry is not the focus of this paper.}
We further assume that the program is constant size and that either the
program is loaded from \permem{} into \locmem{} at restart, or that
there is a small cache for the program itself, which is also lost in
the case of a fault.  Thus, faulting and restarting (loading the base
register and jumping to the restart instruction, and fetching the
code) takes a constant number of external memory transfers.

The processor's computation can be viewed as partitioned into
\emph{capsules}: each capsule corresponds to a maximally contiguous
sequence of instructions running on the processor while the restart
pointer location contains the same restart pointer.
The last instruction of every capsule is therefore a write of
a new restart pointer.  We refer to writing a new restart pointer as
\emph{installing} a capsule.  We assume that the next instructions
after this write, which are at the start of the next capsule, do
exactly the same as a restart does---i.e., load the restart pointer
into the base pointer, load the start instruction pointed to by base
pointer, and jump to it.  The capsule is \emph{active} while its
restart pointer is installed.  Whenever the processor faults, it will
restart using the restart pointer of the active capsule, i.e., the
capsule will be restarted as it was the first time.  We define the
\emph{capsule work} to be the number of external reads and writes in the
capsule, assuming no faults.  Note that, akin to checkpointing,
there is a tension between the desire for high work capsules that
amortize the capsule start/restart overheads and the desire for low
work capsules that lessen the repeated work on restart.

In our analysis, we consider two ways to count the total cost.  We say that
the \emph{faultless work} (or \emph{work}), $W$, is the number of
external memory transfers assuming no faults.  We say that the
\emph{total work} (or \emph{fault-tolerant work}), $\Wf$, is the
number of external transfers for an actual run including all transfers
due to having to restart.  $\Wf$ can only be defined with respect to
an assumed fault model.  In this paper, for analyzing costs, we assume
that the probability of faulting by a processor between any two
consecutive non-zero cost instructions (i.e., external reads or
writes) is bounded by $\faultprob \leq 1/2$, and that faults are
independent events.  We will specify $\faultprob$ to ensure that
a maximum work capsule fails with at most constant probability.

We assume throughout the paper that instructions are
deterministic, i.e., each instruction is a function from the values of
registers and memory locations that it reads to the registers and memory locations
that it writes.

\begin{figure}
  \centering
\includegraphics[width=0.7\columnwidth]{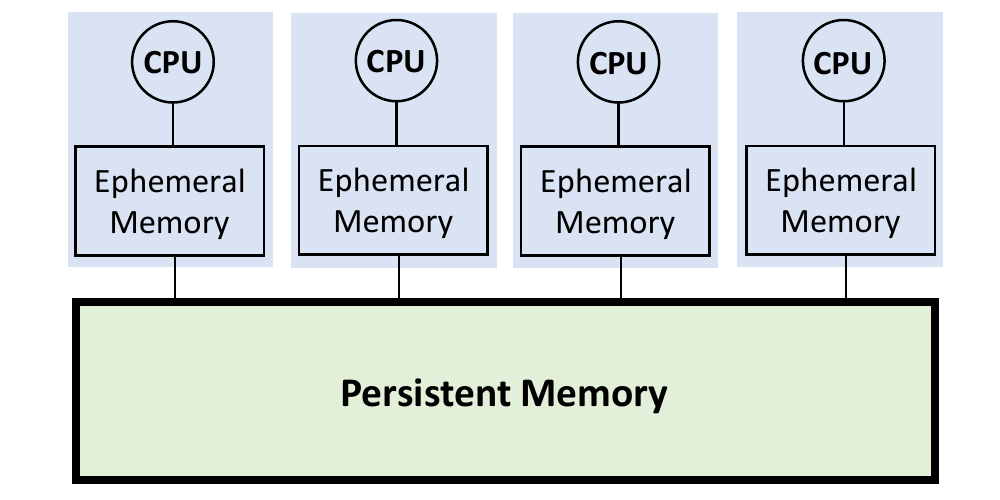}
\vspace{-2pt}
\caption{\label{fig:PPM}
  The Parallel Persistent Memory Model}
\end{figure}

\myparagraph{Multiple Processors} The \ourparallelmodel{} consists of
$P$ processors each with its own fast local \locmem{} of size $M$, but
sharing a single slower \permem{} of size $M_p$ (see Figure~\ref{fig:PPM}).  
Each processor works
as in the single processor \ourmodel{}, and the processors run
asynchronously.  Any processor can fault between two of its
instructions, and each has its own restart pointer location in the
persistent memory.  When a processor faults, the processor restarts
like it would in the single processor case.  We refer to this as a
\emph{soft fault}.  We also allow for a \emph{hard fault}, in which
the processor faults and then never restarts---we say that such a
processor is \emph{dead}.  We assume that other processors can detect
when a processor has hard faulted using a liveness oracle
\texttt{isLive(procId)}.  We allow for concurrent reads and writes to
the shared \permem, and assume that all instructions involving the
\permem{} are sequentially consistent.

The \ourparallelmodel{} includes a compare-and-swap (CAS) instruction.
The CAS takes a pointer to a location of a word in \permem{} and two
values in registers.  If the first value equals the value at the
location, it atomically swaps the value at the location and the value
in the second register, and the CAS is \emph{successful}.
Otherwise, no swap occurs and the CAS is \emph{unsuccessful}.
Even though the \permem{} is organized in
blocks, we assume that the CAS is on a single word within a block.


The (faultless) work
$W$ and the total work $\Wf$ are as defined in the sequential model
but summed across all processors.  The (faultless) \emph{time} $T$
(and the \emph{fault-tolerant} or \emph{total time} $T_f$) is the
maximum faultless work (total work, respectively) done by any
one processor.  Without faults, this is effectively the same as the
parallel external memory model~\cite{Arge08}.
In analyzing correctness, we allow for arbitrary delays between any two
successive instructions by a processor.
However, for our time bounds and our results on work stealing 
we make similar assumptions as made in~\cite{ABP01}.  These
are described in Section~\ref{sec:worksteal}.

\myparagraph{Multithreaded Computations}  Our aim is to support
multithreaded dynamic parallelism layered on top of
the \ourparallelmodel{}.  We consider the same form of multithreaded
computations as considered by Arora, Blumofe, and Plaxton
(ABP)~\cite{ABP01}.  In the model, a computation starts as a single
thread.  On each step, a thread can run an instruction, fork a new
thread, or join with another thread.  Such a computation can be viewed
as a DAG, with an edge between instructions, a pair of out-edges at a
fork, and a pair of in-edges at a join.  As with ABP, we assume that
each node in the DAG has out-degree at most two.  In the
\emph{multithreaded model}, the (faultless) work $W$ is the work
summed across all threads in the absence of faults, and the total work
$\Wf$ is the summed work including faults.  In addition, we define the (faultless)
\emph{depth} $D$ (and the \emph{fault-tolerant} or \emph{total depth}
$D_f$) to be the maximum work (total work, respectively) along any
path in the DAG.  The goal of our work-stealing scheduler
(Section~\ref{sec:worksteal}) is to efficiently map computations in
the multithreaded model into the \ourparallelmodel{}.

%% file: single-proc-robust.tex
\section{Robustness on a Single Processor}\label{sec-single-proc-robust}

In this section, we discuss how to run programs on the single processor
\ourmodel{} model so that they complete the computation properly.

Our goal is to structure the computation and its partitioning into capsules 
in a way that is sufficient to ensure correctness regardless of faults.
Specifically, our goal is that each capsule is a sequence of instructions that
will look from an external view like it has been run exactly once
after its completion, regardless of the number of times it was
partially run due to faults and restarts.  We say that a capsule is
\emph{\idem} if, when it completes, regardless of how many times it faults and restarts, 
all modifications to the \permem{} are consistent with
running once from the initial state (i.e., the state of the \permem{},
the \locmem{}, and the registers at the start of the capsule).

There are various means to guarantee that a capsule is \idem{}, and
here we consider a natural one.  We say that a capsule has a
\emph{\war{} conflict} if the first transfer from a block in
\permem{} is a read (called an ``exposed'' read),
and later there is a write to the same block.
Avoiding such a conflict is important because if a location in the
\permem{} is read and later written, then on restart the capsule would
see the new value instead of the old one.  We say a capsule is
\emph{well-formed} if the first access to each word in the registers
or \locmem{} is a write.  Being well-formed means that a capsule will not
read the undefined values from registers and \locmem{} after a fault.
We say that a capsule is \emph{\war{} conflict free} if it is
well-formed and had no \war{} conflicts.

\begin{theorem}
  \label{theorem:rwcf}
With a single processor, all \war{} conflict free capsules are \idem{}.
\end{theorem}
\begin{proof}
On restarting, the capsule cannot read any \permem{} written by
previous faults on the capsule, because we restart from the beginning
of the capsule and the exposed read locations are disjoint
from the write locations. Moreover, the capsule cannot read the state of the \locmem{}
because a write is required before a read (well-formedness).
Therefore, the first time a
capsule runs and every time a capsule restarts it has the same visible
state, and because the processor instructions are deterministic, will
repeat exactly the same instructions with the same results.
\end{proof}

An immediate question is whether a standard processing model such as
the RAM can be simulated efficiently on the \ourmodel{} model.  The
following theorem shows that the \ourmodel{} can simulate
the RAM model with only constant overheads.

\begin{theorem}\label{theorem:RAM}
Any RAM computation taking $t$ time can be simulated
on the $(O(1),B)$ \ourmodel{} model with $\faultprob \leq 1/c$ for
some constant $c \geq 2$, using $O(t)$ expected total work,
for any $B$ ($B=1$ is sufficient).
\end{theorem}

\begin{proof} 
The simulation keeps all simulated memory in the \permem{} one word
per block.  It also keeps two copies of the registers in \permem{},
and the simulation swaps between the two.  At the end of a capsule on
one copy, it sets the restart pointer to a location just before the
other copy.  The code at that location, run at the start of the next
capsule, copies the other copy of the registers into the current copy,
and then simulates one instruction given by the program counter, by
reading from the other copy of the registers, and writing to the
current copy of registers (typically just a single register).  The
instruction might involve a read or write to the simulated memory, and
an update of the program counter either to the next simulated
instruction, or if a jump, to some other instruction.  Once the
instruction is done, the other copy of the registers is installed.
This repeats.  The capsules are \war{} conflict free because they only
read from one set of registers and write to the other, and the
simulated memory instructions do a single read or write.  Every
simulated step takes a constant number of reads and writes to the
\permem{}.  Since the capsule work is constant, it can be bounded
by some $k$.  If $k \faultprob \leq 1/2$ then the probability
of a capsule faulting is bounded by $1/2$ and therefore the expected
total work on any capsule is upper bounded by $2k$.  Setting $c = 2k$
gives the stated bounds.
\end{proof}

Although the RAM simulation is linear in the number of instructions,
our goal is to create algorithms that require asymptotically fewer
reads and writes to \permem{}.  We therefore consider efficiently
simulating external memory algorithms in the model.

\begin{theorem}\label{theorem:EM}
  Any $(M,B)$ external memory computation with $t$ external accesses
  can be simulated on the $(O(M),B)$ \ourmodel{} model with
  $\faultprob \leq B/(cM)$ for some constant $c \geq 2$, using $O(t)$
  expected total work.
\end{theorem}

\begin{proof}
 The simulation consists of rounds each of which has a
 \emph{simulation} capsule and a \emph{commit} capsule.  It maps the
 \locmem{} of the source program to part of the \locmem{}, and the
 external memory to the \permem.  It keeps the registers in the
 \locmem{}, and keeps space for two copies of the simulated \locmem{}
 and the registers in the \permem, which it swaps back and forth
 between.

 The simulation capsule simulates some number of steps of
 the source program.  It starts by reading in one of the two copies of
 the \locmem{} and registers.  Then during the simulation all
 instructions are applied within their corresponding memories, except
 for writes from the \locmem{} to the \permem.  These
 writes, instead of being written immediately, are buffered in the
 \locmem{}.  This means that all reads from the external memory have to
 first check the buffer.  The simulation also maintains a count of the
 number of reads and writes to the external memory within a capsule.
 When this count reaches $M/B$, the simulation ``closes'' the capsule.  The
 closing is done by writing out the simulated \locmem{}, the
 registers, and the write buffer to \permem{}.  For \locmem{} and
 registers, this is the other copy from the one that is read.  The
 capsule finishes by installing a commit capsule.

 The commit capsule
 reads in the write buffer from the closed capsule to \locmem{}, and
 applies all the writes to their appropriate locations of the
 simulated external memory in the \permem{}.  When the commit capsule
 is done, it installs the next simulation capsule.

This simulation is \war{} conflict free because the only writes
during a simulation capsule are to the copy of \locmem{},
registers, and write buffer.   The write buffer has no conflicts since
it is not read, and the \locmem{} and registers have no conflicts
since they swap back and forth.
There are no conflicts in the
commit capsules because they read from write buffer
and write to the simulated external memory.  The simulation is
therefore \war{} conflict free.

To see the claimed time and space bounds, we note that the \locmem{}
need only be a constant factor bigger than the simulated \locmem{}
because the write buffer can only contain $M$ entries. Each round
requires only $O(M/B)$ reads and writes to the \permem{} because the
simulating capsules only need the stored copy of the \locmem{}, do at
most $M/B$ reads, and then do at most $O(M/B)$ writes to the other stored
copy.  The commit capsule does at most $M/B$ simulated writes, each
requiring a read from and write to the \permem{}.  Because each round
simulates $M/B$ reads and writes to external memory at the cost of
$O(M/B)$ reads and writes to \permem{}, the faultless work across all
capsules is bounded by $O(t)$.  Because the probability that a capsule faults
is bounded by the maximum capsule work, $O(M/B)$, when $\faultprob \leq B/(cM)$, there is a constant $c$ such that the probability of a capsule faulting
is less than 1.   Since the faults are independent, the expected total work is
a constant factor greater than the faultless work, giving
the stated bounds.
\end{proof}

It is also possible to simulate the ideal cache model~\cite{Frigo99}
in the \ourmodel{} model.  The ideal cache model is similar to the
external memory model, but assumes that the fast memory is managed as a
fully associative cache.  It assumes a cache of size $M$ is organized in blocks of size
$B$ and has an optimal replacement policy.  The ideal cache model
makes it possible to design cache-oblivious algorithms~\cite{Frigo99}.
Due to the following result, these algorithms are also efficient in the \ourmodel{} model.

\begin{theorem}\label{theorem:ideal-cache}
Any  $(M,B)$ ideal cache computation with $t$ cache misses
can be simulated on the $(O(M),B)$ \ourmodel{} model
with $\faultprob \leq B/(cM)$ for a constant $c \geq 2$,
using $O(t)$ expected total work.
\end{theorem}

\begin{proof} 
  The simulation is similar to our external memory simulation,
  using rounds consisting of a \emph{simulation} capsule and a
  \emph{commit} capsule.  During each simulation capsule a simulated
  cache of size $2M/B$ blocks is maintained in the \locmem{}.  The
  capsule starts by loading the registers, and with an empty cache.
  During simulation, entries are never evicted, but instead the
  simulation stops when the cache runs out of space, i.e., after $2M/B$
  distinct blocks are accessed.  The capsule then writes out all dirty
  cache lines (together with the corresponding \permem{} address for
  each cache line) to a buffer in \permem{}, saves the registers and
  installs the commit capsule.  The commit capsule reads in the
  buffer, writes out all the dirty cache lines to their correct
  locations, and installs the next simulation capsule.  The simulation
  is \war{} conflict free.

  We now consider the costs of the simulation.  $O(M)$ \locmem{} is
  sufficient to simulate the cache of size $2M$.  The total faultless
  work of a simulation capsule (run once) is bounded by $O(M/B)$
  because we only have $2M/B$ misses before ending, and then have to
  write out at most $2M/B$ dirty cache blocks.  Accounting for faults,
  the total cost is still $O(M/B)$---given constant probability of
  faults.  The size of the cache and cost are similarly bounded for
  the commit capsule.  We now note that over the same simulated
  instructions, the ideal cache will suffer at least $M/B$ cache
  misses.  This is because the simulation round accesses $2M/B$
  distinct locations, and only $M/B$ of them could have been in the
  ideal-cache at the start of the round, in the best case.  The other
  $M/B$ must therefore suffer a miss.  Therefore each simulation round
  simulates what were at least $M/B$ misses in the ideal cache model
  with at most $O(M/B)$ expected cost in the \ourmodel{} model.  As in
  the previous proofs, and since the capsule work is bounded by
  $O(M/B)$, the probability of a capsule faulting can be bounded by
  $1/2$ when $\faultprob \leq B/(cM)$, for some $c$.  Hence the
  expected total work can be bounded by twice the faultless work.
\end{proof}

%% file: prog-robust.tex
\section{Programming for Robustness}

This simulation of the external memory is not completely satisfactory
because its overhead, although constant, could be significant.  It can
be more convenient and certainly more efficient to program directly
for the model.  Here we describe one protocol for this purpose.  It
can greatly reduce the overhead of using the \ourmodel{} model.  It is
also useful in the context of the parallel model.

Our protocol is designed so capsules begin and end at the boundaries
of certain function calls, which we refer to as \emph{persistent
  function calls}.  Non-persistent calls are ephemeral.  We assume
function calls can be marked as persistent or ephemeral, by the user
or possibly compiler.  Once a persistent call is made, the callee will
never revert back further than the call itself, and after a return the
caller will never revert back further than the return.  All persistent
calls require a constant number of external reads and writes on the
call and on the return.  In an ephemeral function call a fault in a
callee can roll back to before the call, and similarly a fault after a
return can roll back to before the return.  All ephemeral calls are
handled completely in the \locmem{} and therefore by themselves do not
require any external reads or writes.   In addition to the persistent
function call we assume a \texttt{commit} command that forces a capsule
boundary at that point.    As with a persistent call, the commit requires
a constant number of external reads and writes.

We assume that all user code between persistent boundaries is
\war{} conflict free, or otherwise idempotent.  This requires a
style of programming in which results are copied instead of
overwritten.  For sequential programs, this increases the space 
requirements of an algorithm by at most a factor of two. Persistent 
counters can be implemented by placing a
commit between reading the old value and writing the new.   In the
algorithms that we describe in Section~\ref{sec:algs}, this style is very natural.

\input{closure}

 
 
 
 





%% file: closure.tex
\subsection{Implementing Persistent Calls}
\label{sec:closures}

The standard stack protocol for function calling is not \war{}
conflict free and therefore cannot be used directly for persistent
function calls.  We therefore describe a simple mechanism based on
closures and continuation passing in functional programming~\cite{Appel89}.  
The convention also
serves to clearly delineate the capsule boundaries, and will be useful
in the discussion of the parallel model.  We then discuss how this can
be implemented on a stack and can be used for loops.

We use a contiguous sequence of memory words, called a \emph{closure},
to represent a capsule.  The restart pointer for the capsule is the
address of the first word of the closure.  A closure consists of an
instruction pointer in the first position (the start instruction),
local state, arguments, and a pointer to another closure, which we
refer to as the \emph{continuation}.  Typically a closure is constant
size and points indirectly (via a pointer) to non-constant sized data, although this
is not required.  Once a closure is filled, it can be installed and
started.  Any faults while it is active will restart it.  Since the
base of the closure is loaded into the base pointer when starting, the
instructions have access to the local state and arguments.  A closure
can be thought of as a stack frame, but need not be allocated in a
stack discipline.  Indeed, as discussed later, allocating in a stack
discipline requires some extra care.  The continuation can be thought
of as a return pointer, except it does not point directly to an
instruction, but rather another closure (perhaps the parent stack
frame), with the instruction to run stored in the first location.

A persistent function call consists of creating two closures, a
continuation closure and a callee closure, and then installing and
starting the callee closure.  The continuation closure corresponds to
what needs to be run when returning from the callee.  It consists of a
pointer to the first instruction to run on return, any local
variables that are needed on return, an empty slot for the return
result of the callee, and the continuation of the current closure.
The callee closure consists of a pointer to the first instruction to
run in the called function, any arguments it needs, and a pointer to
the continuation closure in its continuation.  The return from a
persistent call consists of writing any results into the closure
pointed to by the continuation (the continuation closure), and then
installing and starting the continuation closure.  Note that if the
processor faults after installing the continuation closure, then a
computation will only back up as far as the start of the continuation.
Therefore persistence occurs at the boundaries (in and out) of
persistent function calls.  Because a capsule corresponds to running a
single installed closure, all capsules correspond to the code run
between two persistent function boundaries.  We note that if a
function call is in tail position (i.e., it cannot call another 
function), then the continuation closure is
not necessary, and the continuation pointer of the current active
capsule can be copied directly to the new callee closure before
installing it (this is the standard tail call optimization).

Our calling mechanism is \war{} conflict free.  This is because we
only ever write to a closure when it is being created, and read when
it is being used in a future capsule.  The one exception is writing
results into a closure, but in this case the callee does the writes,
and the caller does the reads after the return and in a different
capsule.
A loop can be made persistent by using tail recursive function calls.
To avoid allocating a new closure for each, the implementation could
use just two closures and swap back and forth between the two.

Closures can be implemented in a stack discipline by allocating both
the callee and continuation closures on the top of the stack, and
popping the callee closure when returning from the called function,
and the continuation closure when returning from the current function.
Standard stack-based conventions, however, will not be \war{} conflict free
because they are based on side-affecting the current stack, e.g. by
changing the value of local variables.  Also the return address is
typically stored in the child (callee) frame.  Here it is important it
is kept in the continuation closure so that the move to a new closure
can be done atomically by swinging the restart-pointer (changing the
start instruction address and base pointer on the same instruction).
A \emph{commit} command can be implemented in the compiler, or by
hand, by putting the code after the commit into a separate function
body, and making a tail call to it.

To implement closures we need memory allocation.  This can be
implemented in various ways in a \war{} conflict free manner.  One
way is for the memory for a capsule to be allocated starting at a base
pointer that is stored in the closure.  Memory is allocated one after
the other, using a pointer kept in local memory (avoiding the need for
a \war{} conflict to persistent memory in order to update it).  In
this way, the allocations are the same addresses in memory each time
the capsule restarts.  At the end of the capsule, the final value of
the pointer is stored in the closure for the next capsule.  For the
\ourparallelmodel, each processor allocates from its own pool of
persistent memory, using this approach.  In the case 
where a processor takes over for a hard-faulting processor, any
allocations while the taking-over processor is executing on behalf of
the faulted processor will be from the pool of the faulted processor.

%% file: multprocs.tex
\section{Robustness on Multiple Processors}
\label{sec:multiproc}

With multiple processors our previous definition of \idem{} is
inadequate since the other processors can read or write \permem{}
locations while a capsule is running.  For example, even though the
final values written by a capsule $c$ might be \idem{}, other
processors can observe intermediate values while $c$ is running and
therefore act differently than if $c$ was run just once.  We therefore
consider a stronger variant of idempotency that in addition to
requiring that its final effects on memory are if it ran once,
requires that it acts as if it ran atomically.  The requirement of
atomicity is not necessary for correctness, but sufficient for what we
need and allows a simple definition.    We give an example of how it can
be relaxed at the end of the section.

More formally we consider the history of a computation, which is an
interleaving of the \permem{} instructions from each of the
processors, and which abides by the sequential semantics of the
memory.  The history includes the additional instructions due to
faults (i.e., it is a complete trace of instructions that actually
happened).  A capsule within a history is \emph{invoked} at the
instruction it is installed and \emph{responds} at the instruction
that installs the next capsule on the processor.  All instructions of
a capsule, and possibly other instructions from other processors, fall
between the invocation and response.

We say that a capsule in a history is \emph{atomically \idem{}} if
\begin{enumerate}
\item (atomic) all its
  instructions can be moved in the history to be adjacent somewhere between
  its invocation and response without violating the memory semantics, and
\item (\idem{}) the instructions are \idem{} at the spot
  they are moved to---i.e., their effect on
  memory is as if the capsule ran just once without fault.
\end{enumerate}

As with a single processor, we now consider conditions under which
capsules are ensured to be \idem{}, in this case atomically. 
Akin to standard definitions of conflict, race, and race free, we
say that two \permem{} instructions on separate processors
\emph{conflict} if they are on the same block and one is a write.  For
a capsule within a history we say that one of its instructions has a
\emph{race} if it conflicts with another instruction that is between the
invocation and response of that capsule.  A capsule in a history is
\emph{race free} if none of its instructions have a race.

\begin{theorem}\label{theorem:rwcf+rf}
  Any capsule that is \war{} conflict free and race free in a history is
  atomically \idem{}.
\end{theorem}

\begin{proof}
  Because the capsule is race free we can move its instructions to be
  adjacent at any point between the invocation and response without
  affecting the memory semantics.  Once moved to that point, the
  idempotence follows from Theorem~\ref{theorem:rwcf} because the
  capsule is \war{} conflict free.
\end{proof}

This property is useful for user code if one can ensure that the
capsules are race free via synchronization.  We use this extensively
in our algorithms.  However the requirement of being race free is
insufficient in general because synchronizations themselves require
races.  In fact the only way to ensure race freedom throughout a
computation would be to have no processor ever write a location that
another processor ever reads or writes.  We therefore consider some
other conditions that are sufficient for atomic idempotence.

\myparagraph{Racy Read Capsule}
We first consider a \emph{racy read capsule}, which 
reads one location from \permem{} and writes its value to another
location in \permem{}.
The capsule can have other instructions, but none of them can
depend on the value that is read.
A racy read capsule is atomically \idem{} if all
its instructions except for the read are race free.  This is true
because we can move all instructions of the capsule, with possible
repeats due to faults, to the position of the last read.  The capsule
will then properly act like the read and write happened just once.
Because races are allowed on the read location, there can be multiple
writes by other processors of different values to the read location,
and different such values can be read anytime the racy read capsule 
is restarted.  However, because the write location is race free,
no other processor can ``witness'' these possible writes
of different values to the
write location.  Thus, the copy capsule is atomically \idem{}.
A copy capsule is a useful primitive for copying from
a volatile location that could be written at any point into a
processor private location that will be stable once copied.  Then when
the processor private location is used in a future capsule, it will
stay the same however many times the capsule faults and restarts.  We
make significant use of this in the work-stealing scheduler.

\myparagraph{Racy Write Capsule}
We also consider a \emph{racy write capsule}, for which the only
instruction with a race is a write instruction to \permem{}, and the instruction
races only with either read instructions or other write instructions,
but not both kinds.  Such a capsule can be shown to be atomically
\idem{}.  In the former case (races only with reads), then in any history,
the value in the write location during the capsule transitions from
an old value to a new value exactly once no matter how many times
the capsule is restarted. Thus, for the purposes of showing
atomicity, we can move all the instructions of the capsule
to immediately before the first read that sees the new value, or to the
end of the capsule if there is no such read.  Although the first time the
new value is written (and read by other processors) may be part of a capsule 
execution that subsequently faulted,
the effect on memory is as if the capsule ran just once without fault
(idempotency).
In the latter case (races only with other writes), 
then if in the history the racy write capsule is the last writer
before the end of the capsule, we can move all the instructions of the 
capsule to the end of the capsule, otherwise we can move all the instructions
to the beginning of the capsule, satisfying atomicity and idempotency.

\myparagraph{Compare-and-Modify (CAM) Instruction} We now consider
idempotency of the CAS instruction.  Recall that we assume that a CAS
is part of the machine model.  We cannot assume the CAS is race free
because the whole purpose of the operation is to act atomically in the
presence of a race.  Unfortunately it seems hard to efficiently
simulate a CAS at the user level when there are faults.  The problem
is that a CAS writes two locations, the two that it swaps.  In the
standard non-faulty model one is local (a register) and therefore the
CAS involves a single shared memory modification and a local register
update.  Unfortunately in the \ourparallelmodel{} model, the processor
could fault immediately before or after the CAS instruction.  On
restart the local register is lost and therefore the information about
whether it succeeded is lost.  Looking at the shared location does not
help since identical CAS instructions from other processors might have
been applied to the location, and the capsule cannot distinguish its
success from their success.


Instead of using a CAS, here we show how to use a weaker
instruction, a \emph{compare-and modify (CAM)}.  A CAM is simply a CAS
for which no subsequent instruction in the capsule reads the local result
(i.e., the swapped value).\footnote{Some CAS instructions in practice return a boolean
to indicate success; in such cases, the boolean cannot be read either.}
Furthermore, we restrict the
usage of a CAM.  For a capsule within a history we say a write $w$
(including a CAS or CAM) to \permem{} is \emph{non-reverting} if no other
conflicting write between $w$ and the capsule's response changes the
value back to its value before $w$. We define a \emph{CAM capsule} as a 
capsule that contains one non-reverting CAM and may contain other \war{} 
conflict free and race free instructions. 

\input{figure-cam-capsule}

\begin{theorem}
  A CAM capsule is atomically \idem{}.
\end{theorem}

\begin{proof}
  Assume that the CAM is non-reverting and all other
  instructions in the capsule are \war{} conflict free and race free.
  Due to faults the CAM can repeat multiple times, but it can only
  succeed in changing the target value at most once.  This is because
  the CAM is non-reverting so once the target value is changed, it
  could not be changed back.  Therefore if the CAM ever succeeds, for the
  purpose of showing atomicity, in the history we move all the
  instructions of the capsule (including the instructions from faulty
  runs) to the point of the successful CAM.  This
  does not affect the memory semantics because none of the other
  instructions have races, and any of the other CAMs were unsuccessful and
  therefore also have no affect on memory.  At the point of the
  successful CAM the capsule acts like it ran once because it is
  \war{} conflict free---other than the CAM, which succeeded just
  once.  If the CAM never succeeds, the capsule is conflict free and
  race free because the CAM did not do any writes, 
  so Theorem~\ref{theorem:rwcf+rf} applies.
\end{proof}

The example CAM capsule in Figure \ref{fig-cam-capsule} shows one of the 
interesting properties of idempotence: unlike transactions or checkpointing, 
earlier runs that faulted can make changes to the memory that are seen or used by 
other processes. Similarly, these earlier runs can affect the results of the 
successful run, as long as the result is equivalent to a non-faulty run. 

A CAM can be used to implement a form of test-and-set in a constant
number of instructions.  In particular, we will assume a location can
either be \emph{unset}, or the value of a process identifier or other unique
identifier.  A process can then use a CAM to conditionally swap such a
location from \emph{unset} to its unique identifier.  The process can then
check if it ``won'' by seeing if its identifier is in the location.
We make heavy use of this in the work-stealing scheduler to atomically
``steal'' a job from another queue.  It can also be used at the join
point of two threads in fork-join parallelism to determine who got 
there last (the one whose CAM from \emph{unset} was unsuccessful) and hence 
needs to run the code after the join.

\myparagraph{Racy Multiread Capsule} It is also possible to design
capsules that are \idem{} without the requirement of atomicity. By way
of example, we discuss the \emph{racy multiread capsule}. This capsule
consists of multiple racy read capsules that have been combined
together into a single capsule.  Concurrent processes may write to
locations that the capsule is reading between reads, which violates
atomicity. Despite this, a racy multiread capsule is \idem{} since the
results of the final successful run of the capsule will overwrite any
results of partial runs.  We make use of the snapshot capsule in the
work-stealing scheduler to reduce the number of capsules required.  It
is not needed for correctness.

%% file: figure-cam-capsule.tex
\begin{figure}

\includegraphics[width=\columnwidth]{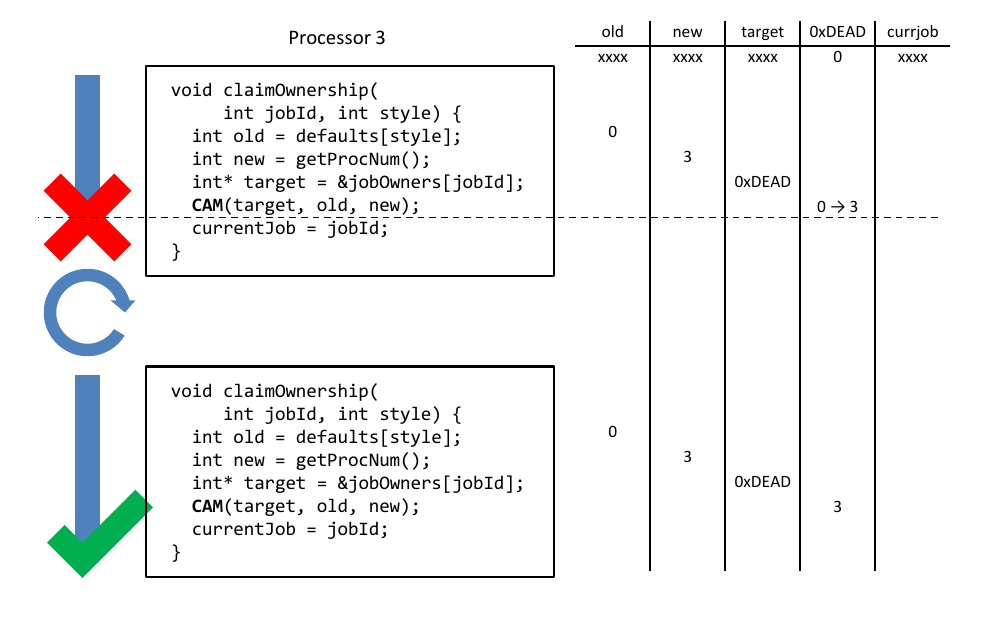}
\vspace{-2em}
\caption{CAM Capsule Example. In CAM capsules, earlier faulting runs of the capsule may perform work that is visible to the rest of the system.}\label{fig-cam-capsule}

\end{figure}

%% file: work-stealing.tex
\section{Work Stealing}\label{sec:worksteal}

We show how to implement an efficient version of work stealing
(WS) in the \ourparallelmodel{} model. Our results are based on the
work-stealing scheduler of Arora, Blumofe, and Plaxton
(ABP)~\cite{ABP01} and therefore work in a multiprogrammed environment
where the number of active processors can change.  As in their
work, we require some assumptions about the machine, which we
summarize here.

The schedule is a two-level scheduler in which the work-stealing
scheduler, under our control, maps threads to processes, and an
adversarial operating system scheduler maps processes to processors.
The OS scheduler can change the number of allocated processors and
which processes are scheduled on those processors during the
computation, perhaps giving processors to other users. The number of
processes and the maximum number of processors used is given by $P$.
The average number that are allocated to the user is $P_{A}$. 

The quanta for scheduling is at least the time for two scheduling
steps where each step takes a small constant number of
instructions. In our case we cannot guarantee that the quanta is big
enough to capture two steps since the processor could fault. However
it is sufficient to show that with constant probability two scheduling
steps complete within the quanta, which we can show.

The available instruction set contains a yield-to-all instruction.
This instruction tells the OS that it must schedule all
other processes that have not hard faulted before (or at the same time) as the process that
executes the instruction. It is used to ensure that processors that
are doing useful work have preference over ones who run out of work
and need to steal.

Our schedule differs from the ABP scheduler in some crucial ways since our model allowing processors to fault.  First, our scheduler cannot use a CAS, for reasons described in Section~\ref{sec:multiproc}, and instead must use a CAM.  ABP uses a CAS and we see no direct translation to using a CAM.  Second, our scheduler has to handle soft faults anywhere in either the scheduler or the user program.  This requires some care to maintain idempotence.  Third, our scheduler has to handle hard faults.  In particular it has to be able to steal from a processor that hard faults while it is running a thread.  It cannot restart the thread from scratch, but needs to start from the previous capsule boundary (a thread can consist of multiple capsules).

Our scheduler is also similar to the ABP scheduler in some crucial ways.  In particular it uses a work-stealing double ended work queue and takes a constant number of instructions for the popTop, popBottom, and pushBottom functions.  This is important in proving the performance bounds and allows us to leverage much of their analysis.  An important difference in the performance analysis is that faults can increase both the total work and the total depth. Because faults can happen anywhere this holds for the user work and for the scheduler.  The expected work is only increased by a constant factor, which is not a serious issue.  However, for total depth, expectations cannot be carried through the maximum implied by parallel execution.  We therefore need to consider high probability bounds.

\subsection{The Scheduler Interface}

For handling faults, and in particular hard faults,
the interaction of the scheduler and threads is slightly different
from that of ABP.  We assume that when a thread finishes it jumps to
the \texttt{scheduler}.\footnote{Note that jumping to a thread is the
  same as installing a capsule.}  When a thread forks another thread,
it calls a \texttt{fork} function, which pushes the new thread on the
bottom of the work queue and returns to the calling thread.  When the
scheduler starts a thread it jumps to it (actually a capsule
representing the code to run for the thread).  Recall that when the
thread is done it jumps back to the scheduler.  These are the only
interactions of threads and the scheduler---i.e. jumping to a thread
from the scheduler, forking a new thread within a thread, and jumping
back to the scheduler from a thread on completion.  All of these occur
at capsule boundaries, but a thread itself can consist of many
capsules.  We assume that at a join (synchronization) point of
threads whichever one arrives last continues the code after the join
and therefore that thread need not interact with the scheduler.  The
other threads that arrive at the join earlier finish and jump to the
scheduler.    In our setup, therefore, a thread is never blocked, assuming the \texttt{fork} function is non-blocking.

\subsection{WS-Deque}

A work-stealing deque (WS-deque) is a concurrent deque supporting a
limited interface. Here we used a similar interface to ABP. In
particular the interface supports popTop, pushBottom, and
popBottom. Any number of concurrent processors can execute popTop, but
only one process can execute either pushBottom or popBottom. The idea
is only the process owning the deque will work on the bottom. The
deque is linearizable except that popTop can return empty even if the
deque is not-empty. However this can only happen if another concurrent
popTop succeeds with a linearization point when the popTop is live,
i.e., from invocation to response.

We provide an implementation of a idempotent WS-deque in
Figure~\ref{WS-deque-fig}.  Our implementation maintains an array of
tagged entries that refer to threads that the processor has either
enabled or stolen while working on the computation.  The tag is simply
a counter that is used to avoid the ABA
problem~\cite{HerlihyShavit2012}.  An \emph{entry} consists of one of
the following states:
\begin{itemize}[noitemsep,topsep=0pt]
\item \emph{empty}: An empty entry is one that has not been associated with a thread yet. Newly created elements in the array are initialized to empty. 
\item \emph{local}: A local entry refers to a thread that is currently being run by the processor that owns this WS-Deque. We need to track local entries to deal with processors that have a hard fault (i.e., never restart).
\item \emph{job}: A job entry is equivalent to the values found in the original implementation of the WS-Deque. It contains a thread (i.e., a capsule to jump to start the thread). 
\item \emph{taken}: A taken entry refers to a thread that has already been or is in the process of being stolen. It contains a pointer to the entry that the thief is using to hold the stolen thread, and the tag of that entry at the time of the steal. 
\end{itemize}
\input{ws-deque-code-atomiclyidempotent}
The transition table for the entry states is shown in Figure~\ref{ws-deque-tran-table}.
\input{ws-deque-transition-table}

In addition to this array of entries, we maintain pointers to the top and the bottom of the deque, which is
a contiguous region of the array.
As new threads are forked by the owner process, new entries will be added to the bottom of the deque using the pushBottom function. 
The bottom pointer will be updated to these new entries. 
The top pointer will move down on the deque as threads are stolen. 
This implementation does not delete elements at the top of the deque, even after steals. 
This means that we do not need to worry about entries being deleted in the process of a steal attempt, but does mean that maintaining $P$ WS-Deques for a computation with span $T_{\infty}$ requires $O(PT_{\infty})$ storage space.  

Our implementation of the WS-Deque maintains a consistent structure that is useful for proving its correctness and efficiency. The elements of our WS-Deque are always ordered from the beginning to the end of the array as follows: 
\begin{enumerate}[leftmargin=*]
\item A non-negative number of taken entries. These entries refer to threads that have been stolen, or possibly in the case of the last taken entry, to a thread that is in the process of being stolen. 
\item A non-negative number of job entries. These entries refer to threads that the process has enabled that have not been stolen or started since their enabling.
\item Zero, one, or two local entries. If a process has one local entry, it is the entry that the process is currently working on. Processes can momentarily have two local entries during the pushBottom function, before the earlier one is changed to a job. If a process has zero local entries, that means the process has completed the execution of its local work and is in the process of acquiring more work through popBottom or stealing, or it is dead.
\item A non-negative number of empty entries. These entries are available to store new threads as they are forked during the computation. 
\end{enumerate}
We can also relate the top and bottom pointers of the WS-Deque (i.e. the range of the deque) to this array structure. The top pointer will point to the last taken entry in the array if a steal is in process. Otherwise, it will point to the first entry after the taken entries. At the end of a capsule, the bottom pointer will point to the local entry if it exists, or the first empty entry after the jobs otherwise. The bottom pointer can also point to the last job in the array or the earlier local entry during a call to pushBottom.

\subsection{Algorithm Overview and Rationale}

We now give an overview and rationale of correctness of our
work-stealing scheduler under the \ourparallelmodel{}.  

\hide{
The high-level overview is that when trying to find a job the
scheduler will first try to take a job from the bottom of its own
deque using popBottom, marking it as \texttt{local}.  If that fails it will
pick a random victim and try to steal from the top of victim's deque
using popTop.  It performs a steal by installing a \texttt{taken} entry in the
victim's deque and a \texttt{local} entry at the bottom of its own deque.  To
get both done without blocking this uses a helper function that other
processes can execute.  If there are no \texttt{jobs} at the victim, but there
is a \texttt{local} on the victim's deque, and the victim is dead, the scheduler
will try stealing the local work (i.e. bring it back to life).  If
there are no jobs and no dead local work the steal attempt fails and
the scheduler tries again by picking another random victim.
}

Each process is initialized with an empty WS-Deque containing enough \texttt{empty} entries to complete the computation. The top and bottom pointers of each WS-Deque are set to the first entry. One process is assigned the root thread. This process installs the first capsule of this thread, and sets its first entry to \texttt{local}. All other processes install the \texttt{findWork} capsule. 

Once computation begins, the adversary chooses processes to schedule according to the rules of the yield instruction described in ABP, with the additional restriction that dead processes cannot be scheduled. When a process is scheduled, it continues running its code. This code may be scheduler code or user code. 

If the process is running user code, this continues until the code calls \texttt{fork} or terminates. Calls to \texttt{fork} result in the newly enabled thread being pushed onto the bottom of the process' WS-Deque. When the user code terminates, the process returns to the \texttt{scheduler} function. 

The scheduler code works to find new threads for the process to work on. It begins by calling the \texttt{popBottom} function to try and find a thread on the owner's WS-Deque. If \texttt{popBottom} finds a thread, the process works on that thread as described above. Otherwise, the process begins to make steal attempts using the \texttt{popTop} function on random victim stacks. 
\hide{
Processes that do not have local work jump to the scheduler. The scheduler begins by clearing the bottom of the running process' WS-Deque to empty. This step ensures that a dead processes will not have a local task to steal.  The scheduler then jumps to the findWork function on the running process' WS-Deque. These steps in combination ensure that each process will only invoke findWork on its own WS-Deque. 

The scheduling loop consists of two parts: a section that handles locally available threads using popBottom and a section that performs steals using popTop. 

PopBottom begins with a capsule that reads the bottom pointer of the WS-Deque and the tag-entry pair above the bottom pointer. This dependency does not lead to concurrency issues since no process other than the owner of the WS-Deque manipulates the bottom pointer. The read of the pair may be out of date, but this is handled later in the function. It is also worth noting that the code assumes the existence of an entry above the bottom. Since the bottom pointer will never decrease below the point it begins the computation at, this assumption is true as long as the bottom pointers are initiated below a zeroth entry. 

The second section of popBottom attempts to claim the pair read in the previous section if it is a job. Otherwise, it returns NULL. In the event of a job, the section performs a CAM to increase the tag and set the entry to local. If the result is as expected, then the process updates the bottom pointer and then returns the claimed continuation. Since WS-Deque entries are only set to local by pushBottom, popBottom, or as a result of a successful steal, we know that if the deque entry is local that it was set during the current capsule or a previous failed attempt. Similarly, we use the fact that the bottom pointer is only modified by pushBottom and popBottom to ensure that there are no races with changing the bottom pointer. If a thief steals the thread after the first section reads it but before the CAM occurs, then the CAM operation will fail and no additional operations will be performed. 

The steal section of findWork loops to make steal attempts until one is successful. Each iteration begins with a yield to prevent thief processes from blocking processes that have useful computation. Once the yield call has been completed, the process chooses a random victim and checks the tag at its bottom pointer. It then enters another capsule to perform the actual steal attempt. If the process fails prior to the steal function, a new victim process will be chosen, but this does not affect the correctness if we assume that the choices are independent. 

The popTop function works to ensure that regardless of faults or concurrency, a positive number of steal attempts that occur simultaneously results in exactly one successful pop. This is primarily done through the use of the helpPopTop function. This function uses the CAM operator to finish a steal that has been made visible. Since all processes running helpPopTop on the victim WS-Deque try to write the same values to the same locations, it does not matter how many failures occur or which process succeeds first. Because the changes are being made with the CAM operator and the tag is non-decreasing, processes that are delayed never overwrite valid data. Similarly, since the top pointer for a WS-Deque never decreases, the CAM on the top pointer of the victim will succeed exactly once without side effects. 

After completing the helpPopTop function, the popTop function has a capsule that checks the top of the victim WS-Deque. The pop is then resolved in the final capsule. If the top of the WS-Deque is empty, then there is nothing to pop and NULL is returned. If the top is taken, that means another thief stole the thread first. The current process will use the helpPopTop function to complete the other thief's pop, then return NULL. If the top is a job, then the active process will attempt to steal that thread using a CAM with a pointer to its own deque. In this case, there are two possibilities. The first is that the CAM succeeds. This means that the steal was successful. After this point, any calls to helpPopTop on this WS-Deque by any process will result in the completion of the steal. Therefore failures do not affect the correctness of this situation. The other possibility is that the CAM fails. This means that a concurrent thief has stolen the job, or that the owner of the WS-Deque started working on the thread. The active process then uses the helpPopTop function to complete a steal if it did occur. Note that if the entry was changed by an invocation of popBottom, the helpPopTop function will have no effect. The tag ensures that once this CAM has failed, it will never succeed regardless of the number of capsule restarts that occur. The final possible state of the top entry is local.  If the owner of the deque is live, the pop is unsuccessful.   Otherwise, the pop follows the same procedure for stealing as in the case of a job, with one exception. Rather than returning the continuation found in the job, the thief will have to extract the continuation from the victim process, including directly reading its restart pointer in order to extract the capsule the victim was working on at its demise.  This prevents the thief from starting capsules that have already been run to completion.

When a thread forks new threads, it calls pushBottom on the running process' WS-Deque. The pushBottom function consists of two separate capsules. The first capsule reads the bottom pointer of the WS-Deque and the tags at and below the bottom. The second capsule sets the new bottom entry to local, sets the bottom of the WS-Deque to the continuation of the job passed in as an argument, and increments the bottom pointer. The writes to the bottom pointer and the entries are independent of reads done during the capsule, so repeating these operations does not change the output. This means that the pushBottom function is repeatable. 

We can also prove that the pushBottom function does not suffer from concurrency issues. We do this by observing that only the processor that owns a WS-Deque will call the pushBottom, popBottom, or clearBottom functions. Since these are the only functions that will edit a local entry or the bottom pointer, we know that no concurrent functions will have conflicts with the changes in Lines 71 or 73 of the code. We rely on the tag to ensure that the CAM to the old bottom pointer happens exactly once. 
}
In a faultless setting, our work-stealing scheduler fuctions like that of ABP. We use the additional information stored in the WS-Deques and the configuration of capsule boundaries to provide fault tolerance. 

We provide correctness in a setting with soft faults using idempotent capsules. Each capsule in the scheduler is an instance of one of the capsules discussed in Section \ref{sec:multiproc}. This means that processes can fault and restart without affecting the correctness of the scheduler. 

Providing correctness in a setting with hard faults is more
challenging. This requires the scheduler to ensure that work being
done by processes that hard fault is picked up in the same capsule
that the fault ocurred during by exactly one other process. We handle
this by allowing thieves to steal \texttt{local} entries from dead
processes. A process can check whether another process is dead using a
liveness oracle \texttt{isLive(procId)}.  

The liveness oracle might be constructed by implementing a counter and
a flag for each process. Each process updates its counter after a
constant number of steps (this does not have to be synchronized). If
the time since a counter has last updated passes some threshold, the
process is considered dead and its flag is set. If the process
restarts, it can notice that it was marked as dead, clear its flag,
and enter the system with a new empty WS-Deque. Constructing such an
oracle does not require a global clock or tight synchronization.

By handling these high level challenges, along with some of the more subtle challenges that occur when trying to provide exactly-once semantics in the face of both soft and hard faults, we reach the following result. 

\begin{theorem}
The implementation of work stealing provided in Figure \ref{WS-deque-fig} correctly schedules work according to the specification in Section \ref{sec:worksteal}.
\end{theorem}

The proof, appearing in Appendix~\ref{sec:WS-proof}, deals with the many possible code interleavings that arise when considering combinations of faulting and concurrency. We discuss our methods for ensuring that work is neither duplicated during capsule retries after soft faults or dropped due to hard faults. In particular, we spend considerable time ensuring that recovery from hard faults during interaction with the bottom of the WS-Deque happens correctly. 

\input{time-bounds}

%% file: ws-deque-code-atomiclyidempotent.tex
\lstset{basicstyle=\footnotesize\ttfamily}
\lstset{escapeinside={@}{@}, literate={<}{{$\langle$}}1  {>}{{$\rangle$}}1}
\lstset{language=C, morekeywords={CAM,commit,empty,local,job,taken,entry,GOTO}}
\lstset{xleftmargin=5.0ex, numbers=left, numberblanklines=false, frame=single}

\begin{figure}[!h]
\begin{lstlisting}[firstnumber=1]
P = number of procs
S  = stack size

struct procState {
  union entry = empty
              | local
              | job of continuation
              | taken of <entry*,int>

  <int,entry> stack[S];
  int top;
  int bot;
  int ownerID;

  inline int getStep(i) { return stack[i].first; }
  
  inline void clearBottom() {
    stack[bot] = <getStep(bot)+1, empty>; }

  void helpPopTop() {     
    int t = top;
    switch(stack[t]) {
      case <_, taken(ps,i)>:
        // Set thief state.
        CAM(ps, <i,empty>, <i+1,local>); @\label{line-help-set-local}@
        CAM(&top, t, t+1);    // Increment top.
   } }

  // Steal from current process, if possible.
  // If a steal happens, location e is set to "local"
  // & a job is returned. Otherwise NULL is returned.
  continuation popTop(entry* e, int c) {
    helpPopTop();
    int i = top;
    <int, entry> old = stack[i];
    commit;
    switch(old) {
      // No jobs to steal and no ongoing local work.
      case <j, empty>: return NULL;
      // Someone else stole in meantime.  Help it.
      case <j, taken(_)>: 
         helpPopTop(); return NULL;
      // Job available, try to steal it with a CAM.
      case <j, job(f)>: @\label{line-steal-job-case}@
        <int, entry> new = <j+1, taken(e,c)>;
        CAM(&stack[i], old, new); @\label{line-steal-job}@
        helpPopTop();
        if (stack[i] != new) return NULL;
        return f; @\label{line-poptop-return-job}@

\end{lstlisting}
\end{figure}

\clearpage

\begin{figure}[!h]
\begin{lstlisting}[firstnumber=50]
     // No jobs to steal, but there is local work.
      case <j, local>: @\label{line-steal-local-case}@
        // Try to steal local work if process is dead.
        if (!isLive(ownerID) && stack[i] == old) {
          commit;
          <int, entry> new = <j+1,taken(e,c)>;
          stack[i+1] = <getStep(i+1)+1, empty>; @\label{line-avoid-extra-steal}@
          CAM(&stack[i], old, new); @\label{line-steal-local}@
          helpPopTop();
          if (stack[i] != new) return NULL;
          return getActiveCapsule(ownerID); @\label{line-poptop-return-local}@
        }
        // Otherwise, return NULL.
        return NULL;
  }  }
  
  void pushBottom(continuation f) {
    int b = bot;
    int t1 = getStep(b+1);
    int t2 = getStep(b);
    commit;
    if (stack[b] == <t2, local>) { @\label{line-pushBottom-if}@
      stack[b+1] = <t1+1, local>; @\label{line-pushBottom-new-local}@
      bot = b + 1; @\label{line-pushBottom-new-bot}@
      CAM(&stack[b], <t2, local>, <t2+1, job(f)> @\label{line-pushBottom-CAM}@
    } else if (stack[b+1].second == empty) {
      states[getProcNum()].pushBottom(f);
    }
    return;
  }
  
  continuation popBottom() {
    int b = bot;
    <int, entry> old = stack[b-1];
    commit;
    if (old == <j, job(f)>) {
      CAM(&stack[b-1], old, <j+1,local>); @\label{line-popBottom-CAM}@
      if (stack[b-1] == <j+1, local>) { @\label{line-popBottom-check-CAM}@
        bot = b-1;
        return f;
      }   }
    // If we fail to grab a job, return NULL.
    return NULL;
  }

  ^ findWork() {
    // Try to take from local stack first.
    continuation f = popBottom();
    if (f) GOTO(f); @\label{line-bottom-pop-GOTO}@
    // If nothing locally, randomly steal.
    while (true) {
      yield();
      int victim = rand(P);
      int i = getStep(bot);
      continuation g 
          = states[victim].popTop(&stack[bot],i); @\label{line-call-popTop}@
      if (g) GOTO(g); @\label{line-top-pop-GOTO}@
} } }
\end{lstlisting}
\end{figure}

\clearpage

\begin{figure}[!h]
\begin{lstlisting}[firstnumber=108]
procState states[P]; // Stack for each process.

// User call to fork.
void fork(continuation f) {@\label{line-fork}@
  // Pushes job onto the correct stack.
  states[getProcNum()].pushBottom(f);
}

// Return to scheduler when any job finishes.
^ scheduler() {@\label{line-scheduler}@
  // Mark the completion of local thread.
  states[getProcNum()].clearBottom();
  // Find work on the correct stack.
  GOTO(states[getProcNum()].findWork());
}
\end{lstlisting}

\caption{Fault-tolerant WS-Deque Implementation.
Jumps are marked as \texttt{GOTO} and
functions that are jumped to and do not return (technically
continuations) are marked with a \texttt{\^}.
All CAM instructions occur in separate capsules, similar to function calls. 
}
\label{WS-deque-fig}
\end{figure}

%% file: ws-deque-transition-table.tex
\begin{figure}[!ht]
\centering
\begin{tabular}{ |c|c|c|c|c|c|  }
\hline
& & \multicolumn{4}{c|}{New State} \\
\hline
& & Empty & Local & Job & Taken\\
\hline
\multirow{4}{4em}{Old State} & Empty & - & \checkmark & &\\
& Local & \checkmark & - & \checkmark & \checkmark\\
& Job & & \checkmark & - & \checkmark\\
& Taken & & & & -\\
 \hline
\end{tabular}

\caption{Entry state transition diagram}\label{ws-deque-tran-table}

\end{figure}

%% file: time-bounds.tex
\subsection{Time Bounds}
\label{sec:timebounds}

We now analyze bounds on runtime based on the work-stealing scheduler under
the assumptions mentioned at the start of the section (scheduled in fixed quanta,
and supporting a yield-to-all instruction).

As with ABP, we consider the total amount of work done by a computation, and the
depth of the computation, also called the critical path length.   In our case
we have $W$, the work assuming no faults, and $\Wf$, the work including
faults.     In algorithm analysis the user analyzes the
first, but in determining the runtime we care about the second.   Similarly we
have both $D$, a depth assuming no faults, and $D_f$, a depth with
faults.

For the time bounds we can leverage the proof of ABP.  In particular
as in their algorithm our popTop, popBottom, and pushBottom functions all take $O(1)$ work without faults.  With our deque, operations take expected $O(1)$ work.  Also as with their version, our popTop is unsuccessful
(returns Null when there is work) only if another popTop is successful
during the attempt.  The one place where their proof breaks down in
our setup is the assumption that a constant sized quanta can always
capture two steal attempts.  Because our processors can fault multiple
times, we cannot guarantee this.  However in their proof this is
needed to show that for every $P$ steal attempts, with probability at
least $1/4$, at least $1/4$ of the non-empty deques are successfully
stolen from (\cite{ABP01}, Lemma 8).  In our case a constant fraction
$(1 - O(1) \cdot \faultprob)^2$ of adjacent pairs of steal attempts will not fault at all
and therefore count as a steal attempt.  For analysis we can assume
that if either steals in a pair faults, then the steal is unsuccessful.  This gives a
similar result, only with a different constant, i.e.,  with probability
at least $1/4$, at least $(1 - O(1) \cdot \faultprob)^2/4$ of the non-empty deques are
successfully stolen from.
We note that hard faults affect the average number of active
processors $P_A$.   However they otherwise have no asymptotic affect in our
bounds because a hard fault in our scheduler is effectively the same as forking
a thread onto the bottom of a work-queue and then finishing.

ABP show that their work-stealing scheduler runs in expected time
$O(W/P_A + DP/P_A)$.
To apply their
results we need to plug in $\Wf$ for $W$ because that is the actual work
done, and $D_f$ for $D$ because that is actual depth.
While bounding $\Wf$ to be within a constant factor of $W$ is straightforward,
bounding
 $D_f$ is trickier because we cannot sum expectations to
get the depth bound (the depth is a maximum over paths lengths).
Instead we show that with some high probability no capsule faults
more than some number of times $l$.  We then simply multiply the depth
by $l$.  By making the probability sufficiently high, we can
pessimistically assume that in the unlikely even that any capsule
faults more than $l$ times then, the depth is as large as the work.
This idea leads to the following theorem.

\begin{theorem}   \label{theorem:time}
  Consider any multithreaded computation with $W$ work, $D$ depth,
  and $C$ maximum capsule work (all assuming no faults) for which all
  capsules are atomically idempotent.
  On the \ourparallelmodel{} with $P$ processors, $P_A$ average number
  of active processors, and fault probability bounded by
  $\faultprob \leq 1/(2C)$, the expected total time $T_f$ for the
  computation is
  \[O\left(\frac{W}{P_A} + D\left(\frac{P}{P_A}\right) \left\lceil\log_{1/(C \faultprob)} W\right\rceil\right).\]
\end{theorem}
\begin{proof}
We must account for faults in both the computation and the
work-stealing scheduler. The work-stealing scheduler has $O(1)$
maximum capsule work, which we assume is at most $C$.  Because we
assume all faults are independent, the probability that a capsule will
run $l$ or more times is upper bounded by $(C\faultprob)^l$.
Therefore if there are $\kappa$ capsules in the computation including
the capsules executed as part of the scheduler, the probability that
any one runs more than $l$ times is upper bounded by
$\kappa(C\faultprob)^l$ (by the union bound).  If we want to bound
this probability by some $\epsilon$, we have $\kappa(C\faultprob)^l
\leq \epsilon$.  Solving for $l$ and using $\kappa \leq 2W$ gives $l
\leq \lceil \log_{1/(C\faultprob)} (2W/\epsilon)\rceil$.  This means
that with probability at most $\epsilon$, $D_f \leq D
\log_{1/(C\faultprob)} (2W/\epsilon)$.  If we set $\epsilon = 2/W$
then $D_f \leq 2 D \log_{1/(C\faultprob)} W$.  Now we assume that if
any capsule faults $l$ times or more that the depth of the computation
equals the work.  This gives $(P/P_A) (2/W) W + (1 - 2/W) 2 D
\lceil\log_{1/(C\faultprob)} W\rceil)$ as the expected value of the
second term of the ABP bound, which is bounded by $O((P/P_A) D
\lceil\log_{1/(C\faultprob)} W\rceil)$.  Because the expected total
work for the first term is $\Wf \leq (1/(1 - C\faultprob)) W$, and
given $C\faultprob \leq 1/2$, the theorem follows.
\end{proof}

This time bound differs from the ABP bound only in the extra
$\log_{1/(C\faultprob)} W$ factor.  If we assume $P_A$ is a constant
fraction of $P$ then the expected time simplifies to $O(W/P + D
\lceil\log_{1/(C\faultprob)} W\rceil)$.

%% file: algorithms.tex
\section{Fault-Tolerant Algorithms}
\label{sec:algs}

In this section, we outline how to implement several algorithms for
the \ourparallelmodel{} model.  The algorithms are all based on binary
fork-join parallelism (i.e., nested parallelism), and hence fit within
the multithreaded model.  We state all results in terms of faultless
work and depth.  The results can be used with
Theorem~\ref{theorem:time} to derive bounds on time for the
\ourparallelmodel{}.  Recall that in the \ourparallelmodel{} model,
external reads and writes are unit cost, and all other instructions
have no cost (accounting for other instructions would not be hard).
The algorithms that we use are already race-free.  Making them
\war{} conflict free simply involves ensuring that reads and
writes are to different locations.  All capsules of the algorithms are
therefore atomically idempodent.
The base case for each of our variants of the
algorithms is done sequentially within the \locmem{}.

\input{prefix-sum}

\input{merge}

\input{sorting}

\input{matrix-multiply}





%% file: prefix-sum.tex
\myparagraph{Prefix Sum}\label{sec:algo-prefixsum}
Given $n$ elements $\{a_1,\cdots,a_n\}$ and an associative operator ``$+$'', the prefix sum algorithm computes a list of prefix sums $\{p_1,\cdots,p_n\}$ such that $p_i=\sum_{j=1}^{i}a_j$.
Prefix sum is one of the most commonly-used building blocks in parallel algorithm design~\cite{JaJa92}.

We note that the standard prefix sum algorithm~\cite{JaJa92} works well in our setting.
The algorithm consists of two phases---the up-sweep phase and the down-sweep phase, both based on divide-and-conquer.
The up-sweep phase bisects the list, computes the sum of each sublist recursively, adds the two partial sums as the sum of the overall list, and stores the sum in the persistent memory.
After the up-sweep phase finishes, we run the down-sweep phase with the same bisection of the list and recursion.
Each recursive call in this phase has a temporary parameter $t$, which is initiated as $0$ for the initial call.
Then within each function, we pass  $t$ to the left recursive call and $t+\mb{LeftSum}$ for the right recursive call, where $\mb{LeftSum}$ is the sum of the left sublist computed from the up-sweep phase.
In both sweeps the recursion stops when the sublist has no more than $B$ elements, and we sequentially process it using $O(1)$ memory transfers.
For the base case in the down-sweep phase, we set the first element $p_i$ to be $t+a_i$, and then sequentially compute the rest of the prefix sums for this block.
The correctness of $p_i$ follows from how $t$ is computed along the path to $a_i$.

This algorithm fits the \ourparallelmodel{} model in a straightforward manner.
We can place the body of each function call (without the recursive calls) in an individual capsule. 
In the up-sweep phase, a capsule reads from two memory locations and stores the sum back to another location.
In the down-sweep phase, it reads from at most one memory location, updates $t$, and passes $t$ to the recursive calls.
Defining capsules in this way provides \war{} conflict-freedom and limits the maximum capsule work to a constant.

\begin{theorem}
The prefix sum of an array of size $n$ can be computed in $O(n/B)$
work, $O(\log n)$ depth, and $O(1)$ maximum capsule work, using
only atomically-idempotent capsules.
\end{theorem}


%% file: merge.tex
\myparagraph{Merging}\label{sec:algo-merge}
A merging algorithm takes the input of two sorted arrays $A$ and $B$ of size $l_A$ and $l_B$ ($l_A+l_B=n$), and returns a sorted array containing the elements in both input lists.
We use an algorithm on the \ourparallelmodel{} model based on the classic divide-and-conquer algorithm~\cite{blelloch2010low}.

The first step of the algorithm is to allocate the output array of size $n$.
Then the algorithm conducts dual binary searches of the arrays in parallel to find the elements ranked $\{n^{2/3}, 2n^{2/3}, 3n^{2/3}, \ldots , (n^{1/3} - 1) n^{2/3}\}$
among the set of keys from both arrays, and recurses on each pair of subarrays until the base case when there are no more than $B$ elements left (and we switch to a sequential version).
We put each of the binary searches into a capsule, as well as each base case.
These capsules are \war{} conflict free because the output of each capsule is written to a different subarray.   Based on the analysis in~\cite{blelloch2010low} we have the following theorem.

\begin{theorem}
Merging two sorted arrays of overall size $n$ can be done in $O(n/B)$ work,
$O(\log n)$ depth, and $O(\log n)$ maximum capsule work, using only
atomically-idempotent capsules.
\end{theorem}


%% file: sorting.tex
\myparagraph{Sorting}
Using the merging algorithm in Section~\ref{sec:algo-merge}, we can
implement a fault-tolerant mergesort with $O((n/B)\log (n/M))$ work
and maximum capsule work $O(\log n)$.  However, this is not optimal.  We now
outline a samplesort algorithm with improved work
$O(n/B\cdot\log_M n)$, based on the algorithm
in~\cite{blelloch2010low}.


The sorting algorithm first splits the set of elements into $\sqrt{n}$
subarrays of size $\sqrt{n}$ and recursively sorts each of the
subarrays.  The recursion terminates when the subarray size is less than $M$,
and the algorithm then sequentially sorts within a single capsule.
Then the algorithm samples every $\log n$'th element from each subarray.
These samples are sorted using mergesort, and $\sqrt{n}$ pivots are picked from the result using a fixed stride.
The next step is to merge each $\sqrt{n}$-size subarray with the
sorted pivots to determine bucket boundaries within each subarray.
Once the subarrays have been split, prefix sums and matrix transposes are used to determine the location in the
buckets where each segment of the subarray is to be sent.  After that,
the keys need to be moved to the buckets, using a bucket
transpose algorithm.  We can use our prefix sum algorithm and the
divide-and-conquer bucket transpose algorithm from~\cite{blelloch2010low}, where the base case is a matrix of size
less than $M$, and in the base case the transpose is done sequentially
within a single capsule (note
that this assumes $M > B^2$ to be efficient).
The last step is to recursively sort the elements within each bucket.
All steps can be made \war{} conflict free by writing to locations separate than those being read.
By applying the analysis in~\cite{blelloch2010low} with the
change that the base cases (for the recursive sort and the
transpose) are when the size fits in the \locmem{}, and that the base
case is done sequentially, we obtain the following theorem.

\begin{theorem}
Sorting $n$ elements can be done in $O(n/B\cdot\log_M{n})$ work,
$O((M/B+\log n)\log_M n)$ depth, and $O(M/B)$ maximum capsule work,
using only atomically-idempotent capsules.
\end{theorem}
It is possible that the $\log n$ term in the depth could be reduced using
a sort by Cole and Ramachandran~\cite{Cole17}.


%% file: matrix-multiply.tex
\myparagraph{Matrix Multiplication}
Consider multiplying two square matrices $A$ and $B$ of size $n\times n$ (assuming $n^2>M$) with
 the standard recursive matrix multiplication~\cite{CLRS} based on the 8-way divide-and conquer approach.
\begin{align*}
   &\left(
    \begin{array}{cc}
      A_{11} & A_{12} \\
      A_{21} & A_{22} \\
    \end{array}
  \right)
  \times
   \left(
    \begin{array}{cc}
      B_{11} & B_{12} \\
      B_{21} & B_{22} \\
    \end{array}
  \right)\\
 =\,&
   \left(
    \begin{array}{cc}
      A_{11}B_{11}+A_{12}B_{21} & A_{11}B_{12}+A_{12}B_{22} \\
      A_{21}B_{11}+A_{22}B_{21} & A_{21}B_{12}+A_{22}B_{22} \\
    \end{array}
  \right)
\end{align*}
Note that every pair of submatrix multiplications shares the same output location.
This leads to \war{} conflicts since a straightforward implementation will read the value from the output cell, add the computed value, and finally write the sum back.
Therefore, the algorithm allocates two copies of temporary space for the output in each recursive subtask, which allows applying computation for the matrix multiplication in two subtasks on different output spaces (with no conflicts), and eventually adding computed values from the temporary space back to the original output space.

If we stack-allocate the memory for each processor, a straightforward upper bound for the total extra storage is $O(pn^2)$ on $p$ processors using the standard space bound under work-stealing.
A more careful analysis can tighten the bound to $O(p^{1/3}n^2)$.
This should be significantly better than the worst-case bound of $\Theta(n^3/(B\sqrt{M})$ when plugging in real-world parameters.
This extra storage can be further limited to $O(n^2)$ by slightly modifying the orders of the recursive calls, assuming the main memory size is larger than the overall size of all private caches. 

When this algorithm is scheduled by a randomized work-stealing scheduler, the whole computation is race-free.
All multiplications that run at the same time have different output locations.
The additions are independent of each other, and
applied after the associated multiplications. Therefore all operations are race-free.

However, if we put each arithmetic operation in a separate capsule, the whole algorithm incurs $O(n^3)$ memory accesses, which is inefficient.
Hence, we mark a capsule anytime the recursion reaches a subtask that can entirely fit into the \locmem{}.
This happens when the matrix size is smaller than $c'\sqrt{M}$ for a constant $c'<1$.
We then continue to run the algorithm sequentially within these capsules.
For the matrix additions, we similarly mark a capsule boundaries such that each capsule can fit into the \locmem{}. This does not affect the overall work.
We obtain the following theorem.

\hide{
Regarding the asymmetric cost between reads and writes, we can use the algorithm in~\cite{BG2018} that improves the I/O cost by a factor of $O(\wcost^{2/3})$.
To achieve the same bound, we mark the base cases as capsules when the output matrix fit into the \locmem{}, and sequentially run the algorithm within each capsule.
The restart cost is $O(\wcost^{4/3}M^{1.5})$ operations, $O(\wcost{}M)$ read transfers and $O(M)$ write transfers.
}

\begin{theorem}
Multiplying two square matrices of size $n$ can be done in $O(n^3/B\sqrt{M})$ work, $O(M^{3/2}+\log^2 n)$ depth, and $O(M^{3/2})$ maximum capsule work, using only atomically-idempotent capsules.
\end{theorem}

We note that we can extend this result to non-square matrices using a similar approach to~\cite{Frigo99}.

\hide{
Note that matrix multiplication is the building block of many other algorithms and problems.
Since this new matrix multiplication approach is resilient to faults, we then have fault-tolerant algorithms with small recovery cost on problems in linear algebra (e.g.\ Gaussian elimination, LU decomposition, triangular solver), all-pair shortest-paths, dynamic programming (e.g.\ the LWS/GAP/RNA/Parenthesis problems), machine learning algorithms based on matrix multiplication, etc.
}

%% file: conclusion.tex
\section{Conclusion}

In this paper, we describe the Parallel Persistent Memory model, which
characterizes faults as loss of data in individual processors and
their associated volatile memory. For this paper, we consider an
external memory model view of algorithm cost, but the model could
easily be adapted to support other traditional cost models. We also
provide a general strategy for designing programs based on capsules
that perform properly when faults occur.  We specify a condition of
being atomically idempotent that is sufficient for correctness, and
provide examples of atomic idempotent capsules that can be used to
generate more complex programs.  We use these capsules to build a
work-stealing scheduler that can run programs in a parallel system
while tolerating both hard and soft faults with only a modest increase
in the total cost of the computation.   We also provide several
algorithms designed to support fault tolerance using our capsule
methodology.
We believe that
the techniques in this paper can provide a practical way to provide
the desirable quality of fault tolerance without requiring significant
changes to hardware or software.

%% file: work-stealing-formal.tex
\section{Proof of the Correctness of Work-Stealing}\label{sec:WS-proof}
Throughout our proof of correctness, we will refer to the code of the work-stealing scheduler shown in Figure~\ref{WS-deque-fig}.
We begin by stating some definitions and assumptions. 
We assume that at least one process will not hard fault during the computation. If this is not true, the computation will have no processes performing work and will never finish. 
The local continuation of a process can be queried using the function getActiveCapsule. This function may be persistent or ephemeral.
Any process can query whether another process has hard faulted through the ephemeral function isDead. 
We define the owner of a WS-Deque to be the process that has the same process number as the ownerID field of that WS-Deque. 
We consider a popBottom to be successful if the CAM at Line~\ref{line-popBottom-CAM} is successful. 
We consider a popTop to be successful if either of the CAM operations at Lines~\ref{line-steal-job} or~\ref{line-steal-local} are successful. 

The first property we prove about our implementation is that the bottom of a WS-Deque can only be operated on by one process at any time. 

\begin{lemma}\label{lemma-deque-bottom-fn}
For a given WS-Deque, only one process can call pushBottom, popBottom, or clearBottom at any time. This process is the owner of the WS-Deque unless the owner hard faults in the middle of a pushBottom or popBottom invocation.
\end{lemma}
\begin{proof}
We first consider the pushBottom function. All calls to pushBottom are made from the fork function. These calls are always made to the WS-Deque chosen by the getProcNum function. Since this function returns the ID of the process that is running it and there is no capsule boundary between the call to getProcNum and the call to pushBottom, the process running the fork will always invoke pushBottom on its own WS-Deque. Since the pushBottom function is part of the scheduler rather than the algorithm code, it is never pushed onto the WS-Deque as a job. This means that it can only be stolen from the local state in the event of a process hard fault. Similarly, all calls to popBottom and clearBottom are made using the getProcNum function inside the findWork and fork functions respectively. Therefore, the same argument holds. Since only the owner can invoke these functions and it will run them to completion before calling any other functions, we know that at most one of these function invocations can exist at any time. 
\end{proof}

From this property, we find the related result. 

\begin{corollary}\label{corollary-deque-bottom-pointer}
For a given WS-Deque, only one process can update the bottom pointer at any time. This process is the owner of the WS-Deque unless the owner hard faults in the middle of a pushBottom or popBottom invocation.
\end{corollary}
\begin{proof}
The only functions that update the bottom pointer are pushBottom, popBottom, or clearBottom. Applying Lemma \ref{lemma-deque-bottom-fn} gives the desired result. 
\end{proof}

We then use this property about the bottom of WS-Deques to show that user level threads that are being worked on by processes are tracked with local entries. 

\begin{lemma}\label{lemma-user-work-implies-bottom-local}
Every process that is working on user level threads will have a local entry that is pointed to by the bottom pointer of their WS-Deque. 
\end{lemma}
\begin{proof}
All user threads are initiated by the findWork function at Line~\ref{line-bottom-pop-GOTO} or Line~\ref{line-top-pop-GOTO}. 

If the thread is started at Line~\ref{line-bottom-pop-GOTO}, it means that popBottom returned that continuation. The if statement at Line~\ref{line-popBottom-check-CAM} requires a local entry to exist at stack[b-1] in order for a non-NULL return value. The bottom pointer is then set to this location before the return. Corollary \ref{corollary-deque-bottom-pointer} tells us that bottom pointer will not be modified by any other process. The entry pointed to by the bottom pointer can only be modified from local by calls to pushBottom or popTop. We know from Lemma \ref{lemma-deque-bottom-fn} that pushBottom cannot be running concurrently. We show that popTop cannot concurrently modify the entry by observing that popTop will only modify a local entry for a process that hard faulted, and a process cannot return a value after it hard faults. Therefore, the values that exist at Line~\ref{line-popBottom-check-CAM} must still exist upon jumping to the continuation. 

If the thread is started at Line~\ref{line-top-pop-GOTO}, it means that popTop returned that continuation. The popTop function can return a non-NULL value at Line~\ref{line-poptop-return-job} or Line~\ref{line-poptop-return-local}. In either case, the return is preceded by a call to the helpPopTop function. This function ensures that the entry pointed to by the newly taken entry is set to local. This newly taken entry was set by the CAM at Line~\ref{line-steal-job} if the return happened at Line~\ref{line-poptop-return-job} or the CAM at Line~\ref{line-steal-local} if the return happened at Line~\ref{line-poptop-return-local}. Both of these CAMs set the entry pointer in the taken to the argument passed to popTop. Looking at Line~\ref{line-call-popTop}, we see that this is the pointer to the bottom of the thief's WS-Deque. Therefore, that is the entry that will be set to local. We know that the bottom entry and pointer will not be modified between the call to helpPopTop and the jump to the continuation because the owner process is the one running the calls to popTop and findWork and the jump to the thread, and can therefore not hard fault or make other calls to pushBottom, clearBottom, or popTop. 

In both cases where user threads are started, a local entry exists on the bottom of the WS-Deque owned by the process starting that thread. It then remains to show that this local entry is not deleted before the process ceases working on that thread. Local entries are only modified by the clearBottom, pushBottom, and popTop functions. We know from Lemma \ref{lemma-deque-bottom-fn} that unless the process hard faults, only the owner can run pushBottom or clearBottom. If the owner calls clearBottom, it must have done so from the scheduler function. This function is only called when the user level thread completes, meaning the process is no longer working on it. Calls to pushBottom may modify the local entry that existed prior to the call if the CAM Line~\ref{line-pushBottom-CAM} succeeds, but Lines~\ref{line-pushBottom-new-local} and~\ref{line-pushBottom-new-bot} will create a local entry at the new bottom before this can happen. Calls to popTop will never modify a local entry unless the owner has hard faulted. In this case, the local entry will be set to taken by the CAM at Line~\ref{line-steal-local}. Once this CAM is successful, the taken entry will point to the bottom of the thief's WS-Deque, which will be an empty entry. The first helpPopTop call on the victim's WS-Deque that resolves Line~\ref{line-help-set-local} will change the empty entry to local. Since the thief must complete a call to helpPopTop between Line~\ref{line-pushBottom-CAM} and the return from popTop, the local entry will be created before the thief begins working on the thread. 
\end{proof}

Since a process can never work on multiple user level threads, we provide a lemma showing that there are never multiple local entries visible to steal if the process crashes. 

\begin{lemma}\label{lemma-one-local-steal}
At most one local entry can be successfully targeted by a popTop on a WS-Deque. Calls to the popTop function of that WS-Deque after the successful steal completes will target an empty entry. 
\end{lemma}
\begin{proof}
In order for a local entry to be stolen the top pointer of the WS-Deque must point to that entry. Since this entry is a local entry, any thief will execute the case beginning at Line~\ref{line-steal-local-case}. In this case Line~\ref{line-avoid-extra-steal} will executed prior to any CAM operation. This will set the entry below the top pointer to empty. Once the local entry has been stolen, the top pointer will be changed to point to the empty entry by the helpPopTop function. No popTop targeting an empty entry will succeed, or perform any modifications to WS-Deque at all. As long as the entry remains empty, no popTop on that WS-Deque can succeed. Empty entries can only be modified by the pushBottom function. The code that performs this modification is enclosed in the if statement at Line~\ref{line-pushBottom-if}. The condition in this if statement will always fail since the CAM at Line~\ref{line-steal-local} removes the remaining local entry from the WS-Deque. Since the empty entry pointed to by the top pointer will never be modified, no further popTop calls can be successful. 
\end{proof}

Having completed these useful structural lemmas, we can begin to prove the correctness of our functions. We focus first on proving correctness in the face of soft faults and leave hard faults for later in the proof. 

\begin{lemma}\label{lemma-popBottom-success}
Any popBottom function targeting a job entry will be successful unless a concurrent popTop function targeting the same entry is successful or the process hard faults. 
\end{lemma}
\begin{proof}
If at any point during the findWork function the process hard faults, then the lemma is vacuously true. This means that we can ignore hard faults for the sake of the proof. 

The entry targeted by a popBottom invocation is the entry immediately above the bottom pointer. If this entry is a job, the CAM at Line~\ref{line-popBottom-CAM} will succeed unless the entry is changed before the CAM happens. Job entries are only modified by successful invocations of popBottom or popTop, so if neither of these functions concurrently succeed on the target entry, the CAM will succeed, and therefore the popBottom will succeed. 
\end{proof}

\begin{lemma}\label{lemma-popTop-success}
Any popTop function targeting a job entry or a local entry on a process that hard faulted will be successful unless a concurrent popBottom or popTop function targeting the same entry is successful or the process hard faults. 
\end{lemma}
\begin{proof}
If at any point during the findWork function the process hard faults, then the lemma is vacuously true. This means that we can ignore hard faults for the sake of the proof. 

We first consider the case when the top pointer points to a job entry. In this case the CAM at Line~\ref{line-steal-job} will succeed unless the entry is changed before the CAM happens. Job entries are only modified by successful invocations of popBottom or popTop, so if neither of these functions concurrently succeed on the target entry, the CAM will succeed, and therefore the popTop will succeed. 

We next consider the case when the victim has hard faulted and the top pointer points to a local entry. In this case, the CAM at Line~\ref{line-steal-local} will succeed unless the entry is changed before the CAM happens. Local entries are only modified by successful invocations of popTop or invocations of pushBottom. We know from Lemma \ref{lemma-deque-bottom-fn} that pushBottom functions can only be run by the owner of the WS-Deque or a thief if the owner of the WS-Deque hard faulted. The owner has hard-faulted, so it cannot run the pushBottom function. Since pushBottom is a scheduler function, it can only be stolen from a local entry, rather than a job entry. By applying Lemma \ref{lemma-one-local-steal} we find that it is impossible for a popTop to target a local entry after the pushBottom function is stolen. By applying Lemma \ref{lemma-one-local-steal} we find that if a popTop invocation targeting a local entry on a process that hard faulted is running concurrently with the pushBottom function for that process' WS-Deque then the entry targeted by the popTop invocation was the target of another successful popTop invocation that ran concurrently with the original popTop invocation. This means that a popTop invocation targeting a local entry on a process that hard faulted will either succeed or be concurrent with another popTop invocation that succeeds at targeting the same entry.

Since we have proven the lemma for both possible cases, the proof is complete.
\end{proof}

When proving the correctness of pushBottom we consider how hard faults affect the implementation of the function and the interleavings that result. We also connect the user level interface function fork to the scheduler. 

\begin{lemma}\label{lemma-push-job}
Every continuation will be added to a WS-Deque as a job exactly the number of times fork was called on it. 
\end{lemma}
\begin{proof}
Job entries are only added to a WS-Deque via the pushBottom function. This function is only ever invoked by the fork function. Each call to fork directly calls pushBottom exactly once. It therefore suffices to show that each call to pushBottom other than recursive calls result in the passed argument being added to a WS-Deque exactly once. We show that each call to pushBottom will have exactly one of the following results: the argument is added to the associated WS-Deque exactly once or the owner hard faults and pushBottom is recursively called with the same argument on a different WS-Deque whose owner has not hard faulted. Since we know that not all processes can hard fault, this is sufficient. 

We begin by assuming that the process does not hard fault while running pushBottom. We know that the bottom entry of the WS-Deque is a local entry by Lemma \ref{lemma-user-work-implies-bottom-local} and that it cannot be concurrently modified by Lemma \ref{lemma-deque-bottom-fn}. This means that all statements inside the if block that begins at Line~\ref{line-pushBottom-if} are executed at least once and that the first execution of the CAM at Line~\ref{line-pushBottom-CAM} will succeed. This adds the continuation to the WS-Deque as a job entry. The tag before the entry prevents the CAM from succeeding more than once. We are then left to show that soft faults will not result in additional calls to pushBottom being made. Until the CAM succeeds, we know that the capsule will always enter the if block at Line~\ref{line-pushBottom-if}. In order for this CAM to succeed, Line~\ref{line-pushBottom-new-local} must be completed, setting the entry below the old bottom pointer to local. Local entries can only be modified by the pushBottom, clearBottom, or popTop functions. The current instance of pushBottom will not change this entry, and Lemma \ref{lemma-deque-bottom-fn} states that no other instance of pushBottom or clearBottom can be running concurrently. We observe that popTop will only modify a local entry on a process that hard faulted. This lets us conclude that the local entry will not be modified if the process does not hard fault, preventing the process from executing the recursive pushBottom call. This means that the continuation is added to the WS-Deque exactly once. 

We then consider the case when the process hard faults while running pushBottom. If the hard fault occurs prior to the CAM at Line~\ref{line-pushBottom-CAM} has succeeded, then the CAM will not succeed on this invocation of pushBottom and the thief that steals this thread will recursively call pushBottom on its own WS-Deque. In this case, the owner hard faulted, so the local entry at stack[b] will not be modified until it is stolen by a call to popTop. During this popTop, the top pointer will point to stack[b]. This means that the thief will set stack[b+1] to empty in Line~\ref{line-avoid-extra-steal} prior to completing the popTop. Furthermore, when the CAM at Line~\ref{line-steal-local} succeeds, it changes the entry at stack[b] from local to taken. When the thief begins runs the pushBottom capsule, it will bypass the if block starting at Line~\ref{line-pushBottom-if} in favor of the else block. Since the if block is not taken, the CAM will never be tried. The else block recursively calls pushBottom with the same argument on the thief's WS-Deque. 

If the hard fault occurs after the CAM succeeds, then the continuation has been added to the WS-Deque and it must not be added again. The tag before the entry prevents the CAM from succeeding more than once. Since the CAM was successful, the entry at stack[b] has been set to job. This means that in order for the pushBottom function to be restarted, a thief had to steal the local entry set at stack[b+1] during Line~\ref{line-pushBottom-new-local}. In order for the steal to occur, the CAM at Line~\ref{line-steal-local} had to succeed, which would change the entry from local to taken. Since taken entries are never modified, we know that stack[b+1] must be a taken entry for the pushBottom function to be resumed. This means that when the thief restarts the capsule, it can never reach the invocation of the pushBottom function. 

We have proven that each call to pushBottom will add the argument to the associated WS-Deque exactly once or recursively call pushBottom with the same argument on a different WS-Deque whose owner has not hard faulted. This proves that each call to fork results in the argument being added to a WS-Deque as a job exactly once, completing the proof. 
\end{proof}

Now that we have shown that user work is correctly added to the scheduler, we show that each process will try to perform the work that has been added. 

\begin{lemma}\label{lemma-work-found}
Every call to findWork results in a successful popTop or a successful popBottom unless the process hard faults or the computation ends. 
\end{lemma}
\begin{proof}
If at any point during the findWork function the process hard faults, then the lemma is vacuously true. This means that we can ignore hard faults by the process calling findWork for the sake of the proof. 

The findWork function begins by calling the popBottom function. Lemma \ref{lemma-popBottom-success} shows that the popBottom call will be successful unless unless all job entries on the WS-Deque are stolen prior to the CAM at Line~\ref{line-popBottom-CAM}. 

If the popBottom function is not successful then the findWork function will proceed to the while loop that performs steal attempts. This loop selects a victim process at random, and the performs the popTop function on that victim. We know from Lemma \ref{lemma-popTop-success} that popTop will succeed if the top entry of the victim's WS-Deque is a job entry or if the victim has crashed and the top entry of its WS-Deque is a local entry unless a concurrent popTop targeting the same entry succeeds. 

In order for the computation to complete, each user thread must be run to completion. This means that if the computation is not complete there is a positive number of user threads that have not been run to completion. Since user threads can only be enabled by other user threads, a at least one of these threads must be enabled. Lemma \ref{lemma-push-job} states that this thread had a job entry created for it. Since the thread has not been completed, it must have a process working on it or its job entry must be in a WS-Deque. If the job entry is in a WS-Deque, then the top pointer of that WS-Deque must point to that entry, or another job entry above it. We know from Lemma \ref{lemma-popTop-success} that if the thief calls popTop on this WS-Deque, it will succeed unless a concurrent call to popBottom or popTop successfully targets that entry. If a process is working on the thread, Lemma \ref{lemma-user-work-implies-bottom-local} states that WS-Deque has a local entry pointed to by the bottom. If the process does not hard fault, it will eventually call fork with a new continuation or complete its user level thread. If fork is called, it will result in a new job entry which may be targeted by popBottom or popTop. If the user level thread is completed, either there exists another enabled user level thread that this analysis applies to, or the computation is complete. If the process hard faults at any time, then its top entry is a valid popTop target and Lemma \ref{lemma-popTop-success} applies. 

At any time, progress is being made towards the end of the computation or there exists a target that the process running findWork may call popBottom or popTop on successfully. Since the findWork function will make attempts to call popTop repeatedly until it succeeds, it will eventually either succeed, or the computation will finish. 
\end{proof}

We extend the proof of processor effort to show that between all of the available processes, all of the work that is added to the scheduler is found without inadvertently duplicating any of that work.

\begin{lemma}\label{lemma-pop-job}
Each job entry in a WS-Deque will be the target of a successful popBottom or popTop exactly once.
\end{lemma}
\begin{proof}
To show that a job entry cannot be successfully popBottomed or popTopped more than once, we note that either successful function is caused by a successful CAM operation on the associated entry. Such a CAM changes the entry to either local or taken depending on which function was successful. 

In order for the computation to complete, each user thread must be run to completion. This means that while there are job entries in any WS-Deque, the scheduler will continue to run. Job entries are located in a WS-Deque above the bottom pointer of the WS-Deque and at or below the top pointer of the WS-Deque. The structure of a WS-Deque means that if a job entry is not directly above the bottom pointer, all entries between it and the bottom entry are job entries. Similarly, if a job entry is not at the top pointer, then all entries between it and the top entry are job entries. 

If a process that does not have user level work on its WS-Deque does not hard fault, it will perform one failed popBottom call, and then repeatedly make popTop attempts until one is successful and it becomes a process that has user level work. These popTop attempts are made on random processes, ensuring that each process will eventually e chosen as a victim. 

If a process that has user level work on its WS-Deque does not hard fault, it will repeatedly run any local work that it has, then call the popBottom function. Lemma \ref{lemma-popBottom-success} tells us that this function will succeed unless it is targeting a non-job entry or a concurrent popTop call targeting the same entry succeeds. If popBottom succeeds, the bottom pointer is set to the next lowest entry and the process is repeated. If popBottom is not targeting a job entry, the structure of a WS-Deque tells us that there are no job entries on that WS-Deque. If a concurrent popTop call targeting the same entry succeeds that entry becomes taken and the top pointer of the WS-Deque must point to that entry. This also means that the WS-Deque contains no job entries. Once the WS-Deque contains no job entries and the local entry (if any) finishes, the process becomes a process that does not have user level work. 

If a process that has user level work on its WS-Deque does hard fault, it will have some non-negative number of job entries above the bottom pointer of its WS-Deque. We rely on thief processes to pop these entries from the WS-Deque, and assume that they exist. Lemma \ref{lemma-popTop-success} tells us that every popTop attempt on the WS-Deque will be successful unless a concurrent popTop or popBottom is successful. Lemma \ref{lemma-deque-bottom-fn} states that popBottom cannot be run on a WS-Deque owned by a process that hard faulted unless it is stolen. Since popBottom is a scheduler function, it can only be stolen by a local entry. The structures of the WS-Deque means that local entries cannot be stolen until there are no job entries on the WS-Deque. This means that the top entry will be the target of a successful popTop call from a thief. Since a successful popTop call results in the top pointer being lowered, this process will repeat until all job entries have been targeted by successful popTop calls.

As long as there exists at least one process that does not hard fault, it will switch between having user level work that it completes to not having user level work and making popTop attempts until it finds some. The end result of this process is that no job entries will remain in any WS-Deque. Since job entries are only modified by successful calls to popBottom or popTop, each job entry must have been removed by a successful popBottom or popTop call. 
\end{proof}

Once the scheduler has assigned threads to various processes, the processes must complete the work. The following two lemmas show that each thread that is assigned has computation begun on it, which is sufficient to show completion in the face of soft faults. 

\begin{lemma}\label{lemma-run-popped-top}
Every continuation that is successfully popTopped is jumped to at least once.
\end{lemma}
\begin{proof}
We consider a continuation to successfully popTopped if the WS-Deque entry associated with that continuation is targeted by a successful CAM operation inside of the popTop function. We consider an entry and a continuation to be associated if the entry is a job containing the continuation or the entry is local while the continuation is being run by the process that owns the WS-Deque the entry resides in. After the successful CAM, the target entry has been set to a taken entry that contains a pointer to the bottom of the thief. Since taken entries are never changed, we know that the if statement will succeed if and only if the CAM succeeded during the current capsule. Therefore if the process does not hard fault then the continuation will be returned to findWork, which jumps to that continuation. Soft faults may cause some of the instructions to be re-run, but will not change the resulting memory state. If the process hard faults at any point between the successful CAM and the jump, it relies on other thieves calling the helpPopTop function to ensure that there is a local entry at the location pointed to in the taken entry, which is the bottom of its WS-Deque. This entry will eventually be stolen by some other thief. That thief will restart the capsule that the original thief hard faulted during. Since we know that not all processes hard fault, at some point a process will complete the popTop function and jump to the continuation inside the findWork function. 
\end{proof}

\begin{lemma}\label{lemma-run-popped-bottom}
Every continuation that is successfully popBottomed is jumped to at least once.
\end{lemma}
\begin{proof}
We consider a continuation to successfully popBottomed if the WS-Deque job entry containing that continuation is targeted by a successful CAM operation inside of the popBottom function. After the CAM is successful, the target entry has been set to local. This local entry can only be changed by a call to clearBottom, or a call to popTop after the process hard faults. By Lemma \ref{lemma-deque-bottom-fn}, we know that clearBottom cannot run concurrently with popBottom. This means that the if statement will succeed if and only if the CAM succeeded during the current capsule. Therefore if the process does not hard fault then the continuation will be returned to findWork, which jumps to that continuation. Soft faults may cause some of the instructions to be re-run, but will not change the resulting memory state. If the process hard faults at any point between the successful CAM and the jump, then then a local entry will exist at the bottom of its WS-Deque until that entry is stolen. Lemma \ref{lemma-run-popped-top} shows that once the entry is stolen, it will be jumped to at least once. Jumping to either popBottom or findWork during the specified window will maintain the local variables, including the continuation that will then be jumped to in findWork. 
\end{proof}

We now show that hard faults do not prevent any computation from being completed. 

\begin{lemma}\label{lemma-restart-work}
Any user thread on a process that hard faulted will be be the result of a successful popTop and the capsule that was in process will be restarted. 
\end{lemma}
\begin{proof}
In order for the computation to complete, each user thread must be run to completion. This means that while there are unfinished user threads, the scheduler will continue to run. Lemma \ref{lemma-user-work-implies-bottom-local} states that any process that is working on a user level thread has a local entry that is pointed to by its bottom pointer. Lemma \ref{lemma-work-found} states that processes that run out of work will eventually perform a successful popBottom or popTop unless they hard fault or the computation ends. Using Lemma \ref{lemma-pop-job}, we know that the number of successful calls that target a job entry is limited. Since successful popBottoms can only target job entries and successful popTops can only target job or local entries, all other successful popTop calls must occur on user threads on processes that have hard faulted. We combine Lemma \ref{lemma-run-popped-top} with the fact that popTop calls getActiveCapsule when stealing a local entry to finish the proof. 
\end{proof}

We have shown that all work created at the user level is completed and that no user level threads are created by the scheduler. We conclude the proof by showing that the scheduler does not over-execute user threads. 

\begin{lemma}\label{lemma-no-duplicate-work}
No capsule in user level code will be run to completion more times than the number of times it is invoked by user level code. 
\end{lemma}
\begin{proof}
A capsule is considered run to completion when all of its instructions have been completed and the restart pointer for the subsequent capsule has been installed. We assume that capsules are handled as discussed in Section {sec-single-proc-robust}, which describes how to ensure that soft faults during direct runs will not cause a capsule to run to completion multiple times. We are then left to show that the scheduler does not result in extra invocations of user level code. This might happen in two ways: threads might be added to WS-Deques more times than they were enabled or entries on WS-Deques may be run multiple times. 

We first show that threads are not added to WS-Deques more times than they are enabled. Threads are only enabled as job entries through calls to the fork function. We show in Lemma \ref{lemma-push-job} that the number of job entries added for a thread is exactly the number of times fork is invoked on that thread. 

We next show that although WS-Deque entries can be run multiple times, this will not result in any capsule being run to completion multiple times. Lemma \ref{lemma-pop-job} states that each job entry will be the result of a successful popBottom or popTop exactly once. In Lemmas \ref{lemma-run-popped-bottom} and \ref{lemma-run-popped-top}, we prove that in either of these cases we will jump to the beginning of the thread. 

If the thread is run to completion without hard faulting, it will complete the user level work normally, possibly make calls to fork, and then call the scheduler function. We know that local work is not stolen from a process unless that process hard faulted, so we do not have to consider steal attempts at this time. The user level work does not interact with the scheduler, and therefore cannot affect the entries on the WS-Deque. If pushBottom is run to completion without any hard faults, then the original entry that corresponded to the user thread will be replaced with a job entry containing the newly enabled thread continuation. A new local entry corresponding to the thread will be created below the original entry. The scheduler function calls clearBottom, sets the local entry at the bottom of the WS-Deque to empty. Once this has been completed there is no longer an entry corresponding to the thread, so it cannot be jumped to again unless it is later re-enabled. Since the process has not hard faulted, no thief will ever steal the local entry associated with the thread. 

We then consider what happens if the process running the thread hard faults. If the process hard faults during the user level code, then it will be stolen regularly. Since the popTop function returns the active capsule when stealing a local entry (Line~\ref{line-poptop-return-local}) rather than the entire thread, the thief will start on the first capsule that has not been run to completion rather the beginning of the thread. The thief will also set the local entry that corresponded to the thread to empty during Line~\ref{line-steal-local}, preventing it from being stolen by any other thief. These facts are true for any steal on a process that has hard faulted. We consider two cases for a hard fault during pushBottom: before the CAM at Line~\ref{line-pushBottom-CAM} is run at all, and after it has been run at least once. If the hard fault occurs before the CAM is run then the entry at stack[b] will remain local until it is stolen. During the popTop call when this entry is stolen, the thief will set stack[b+1] to empty during Line~\ref{line-avoid-extra-steal}. Since this occurs before the CAM at Line~\ref{line-steal-local}, it will occur before the top pointer can be changed to stack[b+1]. Since this entry will be set to empty before any popTop calls can see it, it will never be stolen. When the thief restarts the active capsule in pushBottom, the entry at stack[b] will have been set to taken and the entry at stack[b+1] will have been set to empty, so the thief will call pushBottom on its own WS-Deque. That call can be analyzed in the same manner as the original call. If the hard fault occurs after the CAM at Line~\ref{line-pushBottom-CAM} has been run then the entry at stack[b] has been set to job. In this case, the current thread will not be stolen until the top pointer is set to point to stack[b+1]. This entry was set to local by Line~\ref{line-pushBottom-new-local}. When the thief restarts the active capsule, the state of the WS-Deque will cause it to bypass both if clauses and immediately return to the user thread without further modifying the WS-Deque. If the process hard faults during the scheduler function, the hard fault will either occur before clearBottom finishes, in which case it will be stolen and restarted normally, or it will occur after the clearBottom finishes, in which case the entry will be set to empty and therefore never stolen. 

We have now proven that threads are not added to WS-Deques more times than they are enabled and that WS-Deque entries being run multiple times will not cause a user level capsule inside threads associated with those entries to be run to completion multiple times. This completes the proof. 
\end{proof}

Combining the various lemmas gives the following theorem.

\begin{theorem}
The implementation of work stealing provided in Figure \ref{WS-deque-fig} correctly schedules work according to the specification in Section \ref{sec:worksteal}.
\end{theorem}
\begin{proof}
We know from Lemma \ref{lemma-push-job} and Lemma \ref{lemma-pop-job} that every enabled user thread will be scheduled on to an active process. Lemma \ref{lemma-run-popped-top}, Lemma \ref{lemma-run-popped-bottom}, and Lemma \ref{lemma-restart-work} combine to prove that every scheduled thread is run to completion. Lemma \ref{lemma-pop-job} and Lemma \ref{lemma-no-duplicate-work} show that no work is duplicated or re-executed. Since all work is scheduled and run to completion following the computation dependencies, the implementation is correct. 
\end{proof}